\documentclass[12pt]{amsart}
\usepackage[utf8]{inputenc}
\usepackage[margin=1in]{geometry}
\usepackage{natbib}
\usepackage{multirow}
\usepackage{rotating}
\usepackage{verbatim}
\usepackage{setspace}
\usepackage[table]{xcolor}
\usepackage{colortbl}
\usepackage{hyperref}
\usepackage{subcaption,graphicx,float,pdfpages,tikz}
\usepackage{pgf}
\usepackage{cutwin}
\usetikzlibrary{arrows,automata}
\usepackage{amsmath,amssymb,enumitem,amsthm}
    \theoremstyle{definition}
    
    \newtheorem{assumption}{Assumption}
    \newtheorem{definition}{Definition}
    \theoremstyle{plain}
    \newtheorem{proposition}{Proposition}
    
    \newtheorem{result}{Result}
    \DeclareMathOperator*{\argmax}{arg\,max}

\newcommand{\inv}{^{-1}}
\usepackage{soul,color,xcolor}
\title[Strategic formation of collaborative networks]{Strategic formation of collaborative networks}
\author[Philip Solimine \and Luke Boosey]{Philip Solimine$^1$ \and Luke Boosey$^2$}
\date{\today}
\thanks{
The authors would like to thank Christopher Brown, Arun Chandrasekhar, Krishna Dasaratha, Dean Eckles, Matthew Gentry, Ben Golub, Mark Isaac, Matthew Jackson, Annalise Maillet, Angelo Mele, Luca Paolo Merlino, Anke Meyer-Baese, Svetlana Pevnitskaya, Tim Rooney, Jonathan Stewart, and Cynthia Fan Yang for their advice and feedback, valuable discussions at various stages of production, and/or comments on earlier drafts. Thanks also to seminar participants at UBC, ESA Job Market Candidates Seminar, Conference(s) of Network Science in Economics, Conference(s) of the Southern Economic Association, NetSci satellite on Statistical Physics of Financial and Economic Networks, XS/FS and Quantitative Methods working groups at FSU, and the FSU Computational Exposition. Finally, we are grateful to the Charles \& Persis Rockwood Fellowship, the L. Charles Hilton Center for Economic Prosperity and Individual Opportunity, and the FSU College of Social Sciences \& Public Policy for their generous support. This paper has previously been presented in working form under the title ``Resource sharing on endogenous networks''.\\ 
$^1$ Vancouver School of Economics, University of British Columbia \\\textit{Correspondence:} philip.solimine@ubc.ca \\
$^2$ Department of Economics, Florida State University}

\begin{document}
\begin{abstract}
   We examine behavior in an experimental collaboration game that incorporates endogenous network formation. The environment is modeled as a generalization of the voluntary contributions mechanism. By varying the information structure in a controlled laboratory experiment, we examine the underlying mechanisms of reciprocity that generate emergent patterns in linking and contribution decisions. Providing players more detailed information about the sharing behavior of others drastically increases efficiency, and positively affects a number of other key outcomes. To understand the driving causes of these changes in behavior we develop and estimate a structural model for actions and small network panels and identify how social preferences affect behavior. We find that the treatment reduces altruism but stimulates reciprocity, helping players coordinate to reach mutually beneficial outcomes. In a set of counterfactual simulations, we show that increasing trust in the community would encourage higher average contributions at the cost of mildly increased free-riding. Increasing overall reciprocity greatly increases collaborative behavior when there is limited information but can backfire in the treatment, suggesting that negative reciprocity and punishment can reduce efficiency. The largest returns would come from an intervention that drives players away from negative and toward positive reciprocity.
\end{abstract}

\maketitle

\clearpage
\section{Introduction}

 We study a voluntary collaboration game in which players choose how much, and with whom, to share. In this setting, shared resources generate externalities for other players, but the player sharing the resources does not derive any immediate or direct benefit from doing so. Viewing the decisions about who to share with as a process of endogenous network formation, we estimate a structural model of preferences and use it to infer the value of different social preferences for trust and reciprocity.

Consider, for example, the time allocation problem faced by a researcher involved in multiple projects with different sets of coauthors. The researcher has a limited amount of time and concentration power to dedicate towards coauthored projects and her own research activity. Allocating attention to coauthored projects benefits the coauthors, but effort is a scarce resource, and spreading effective contributions across multiple projects reduces the impact on each project individually. By creating efficiency gains, this type of collaboration can play a valuable role in the economic systems in which it is embedded. Apart from scientific collaboration, a number of relevant real-world examples fit into the basic framework of collaboration with congestion. On the internet, web pages make linking decisions; they are paid advertising revenue for traffic, and they are indexed by a search engine that favors well-linked pages, such as Google's PageRank algorithm \citep{page1999pagerank}. Linking to another site diverts a fraction of traffic and amplifies search appearances for the linked site. Another example concerns the networks of firms who form links by collaborating on R\&D projects \citep[see, e.g.][]{dasaratha2023innovation}, and determine how much intellectual property they will share with other firms through this process.

A pervasive feature of these collaborative environments, in both digital and analog settings, is that a player who invests their effort in generating externalities can often choose where to direct them. This pattern of interaction generates a network whose structure can be informative of the process that governs its formation and evolution. Understanding how reputation and information affect the evolutionary behavior of interconnected social systems can help inform platform designers and policymakers on ways to design digital platforms that promote efficiency.

While the examples above are distinct from each other in several ways, they each highlight some key common characteristics of voluntary sharing in collaborative environments. First, the decision to share naturally entails some individual cost to the sharer. Second, shared resources are, in many cases, congestible. They congest in the sense that sharing the same resources across multiple recipients reduces the effective benefit derived by each recipient. Third, sharing is conventionally associated with altruistic or reciprocal motivations---for example, webpage linking decisions may be forged reciprocally; while firms and coauthors may be more inclined to dedicate their time to a collaborative project if they notice that their fellow researchers are similarly dedicated.

In such settings, trust and reciprocity play key roles in helping players coordinate to reach efficient outcomes \citep{fehr1997reciprocity,fehr1998reciprocity}. People with social preferences tend to enjoy helping others who have helped them. By controlling the specific structure and content of the information a player has about their neighbors' actions, we aim to reason about how changes in these social preferences can be identified by the changing patterns of collaborative behavior.

In order to examine the underlying influences on voluntary resource sharing, we designed an experiment to study a simple sharing environment that incorporates the three key features outlined above. In the experiment, subjects were randomly assigned to groups of four players. Groups interacted with each other for 15 rounds (i.e., without random rematching between rounds). In each round, players privately and simultaneously made two choices; (i) how much of their endowment (20 tokens) to share (referred to as their contribution), and (ii) with whom to share, forming a network. Players retained any endowment not shared. Player $i$'s contribution (her shared resources) generated benefits for her selected recipients and for player $i$ herself; however, the shared resources were congestible, so that each received only some fraction of a token for each contributed token shared by player $i$. More specifically, every contributed token was multiplied by a factor of 1.6, with the resulting proceeds shared equally between player $i$ and her selected recipients.\footnote{This specification of the benefits from shared resources is, by design, comparable to the widely adopted implementation of the \textit{marginal per capita return} (MPCR) in VCM or linear public goods experiments.} At the end of the round, player $i$'s payoff was calculated as the sum of her retained endowment, the return on her own shared resources, and the incoming benefits derived from others who selected player $i$ as a recipient of their own shared resources.

Our experiment allows us to examine the impact of the information structure on reciprocal behavior. In the control (or baseline) condition, players are given information only about the total inflow of benefits from others after each round, but cannot identify the source of those benefits. In this way, direct reciprocity is precluded, although sharing behavior may nevertheless be driven in part by altruism or generalized reciprocity. In the information treatment, players are informed about the incoming benefits from each of the other three group members. In the treatment, they can specifically identify which other group members shared with them, and how much. This enables direct reciprocity to influence both contribution and network formation decisions. In each session of the experiment, groups first played 15 rounds of this collaboration game in the control condition with limited information. Then, they were informed that they would play another 15 rounds with the same group (albeit with shuffled IDs). In control sessions, the second part simply repeated the control condition for 15 more rounds. In the treatment sessions, these rounds included the additional information.

By virtue of the endogenous linking decisions, our study contributes to the growing literature on network formation games, and especially to more recent work examining simultaneous action and link formation decisions. Network formation problems are notoriously difficult to analyze, but have far-reaching consequences and applications in economics and related fields alike.\footnote{See \cite{chandrasekhar2016econometrics} for a comprehensive introduction to the study of network formation.} Fundamentally, the observation of a panel of small networks over time, rather than that of a single large cross-section, allows us to adopt a more sophisticated and innovative structural approach. These methods build on the recent techniques developed by \cite{badev2021nash} and \cite{mele2017astructural} for large cross-sectional observations of social and economic networks. More specifically, panel data allows us to estimate the discrete choice process directly, avoiding the computational issues associated with large-network cross-sectional estimation. As such, one of our main contributions is methodological in nature---we implement a new approach \citep[see, e.g.,][]{overgoor2019choosing,gupta2022mixed} to estimate the network formation game as a combination of individual evolutionary discrete choice strategies.

Our other main contributions are more conceptual in nature. The voluntary sharing environment we study is novel, but it exhibits some similarities with other well-studied strategic decision settings, including dictator giving, the provision of public goods or club goods, and public goods games played on networks. As such, our study is closely related to the literature examining network public goods games \citep{bramoulle2007public,elliott2019network,boosey2017conditional}, and especially to work that extends the game to incorporate endogenous linking decisions \citep{galeotti2010law,kinateder2017public,kinateder2021evolution}. Importantly, while existing literature has focused on settings in which links describe the decision to access another player's provision of the good, our environment captures the reverse situation---the links formed by each player describe the decision to \emph{grant access} to the benefits generated by her own contribution.

Our reduced-form results alone point to striking effects of the information treatment. We find that when subjects can observe who shares with them, as well as how much they share, contributions and the average number of links increase substantially. Fewer nodes are isolated, facilitating the formation of efficient structure, and contribution profiles are more balanced (decentralized) across group members. As one might expect, the availability of more detailed information leads to contribution and linking patterns that exhibit dramatically higher levels of reciprocal behavior than in the baseline condition.

Building off these stark reduced-form results, we develop and estimate a structural model that incorporates both generalized reciprocity and direct reciprocity as core behavioral considerations for the participants. Our estimates suggest that in the baseline condition (without specific information), sharing behavior is driven partly by generalized reciprocity. By contrast, when subjects are provided with more specific information, sharing behavior is heavily influenced by direct reciprocity, and subjects substitute sharply away from generalized reciprocity concerns in favor of this more focused form of reciprocity.

The estimates of our structural parameters are robust to the inclusion of several heterogeneous individual characteristics, which we construct based on subjects' responses to a post-experiment survey. We distill the responses down to three key attributes using a principal components analysis, representing trust, overall reciprocity, and positive reciprocity. Incorporating these heterogeneous characteristics into our estimations, we observe clear and intuitive effects of heterogeneity on behavior. In the baseline condition, higher scores on the trust attribute are consistent with trustworthy behavior (stronger effects of generalized reciprocity), but predict less trusting decisions as subjects tend to be more cautious about sharing altruistically, perhaps out of mistrust in the notion that others will reciprocate. In the treatment condition, knowing that others observe the source of incoming benefits allows subjects with higher trust attributes to be more trusting.

The two principal components of reciprocity---overall reciprocity and positive reciprocity---also explain much of the variation in behavioral patterns. Overall reciprocity, which captures both a taste for positive and negative reciprocity, has a nuanced effect in the two conditions. In the baseline, subjects with a higher overall reciprocity attribute are more altruistic and exhibit a stronger preference for generalized reciprocity. In the treatment condition, they shift more weight towards direct reciprocity, as expected, but also exhibit less altruism and less concern for generalized reciprocity. One interpretation for these differences is that in the treatment condition, the additional information about sharing decisions of others facilitates negative reciprocity (punishment of those who do not share) more than it enhances positive reciprocity.

Subjects with higher scores on the positive reciprocity attribute are also more altruistic in the baseline condition. However, in contrast with our results for the overall reciprocity attribute, in the treatment condition, they share generously with the entire group in hopes of stimulating the reciprocity of others, rather than resorting to punishment out of negative reciprocity.

Finally, we use our structural framework to conduct three counterfactual simulations, each examining the effects of a uniform increase in one of the three principal attributes---trust, overall reciprocity, and positive reciprocity. Consistent with our estimates for the model with individual heterogeneity, an increase in the trust attribute improves key outcomes in the information treatment, but has a mild, negative impact on the baseline condition. Overall reciprocity improves outcomes in the baseline condition but can backfire substantially due to negative reciprocity in the treatment condition. However, a uniform increase in subjects' positive reciprocity attribute has a substantial positive effect on all key outcomes in both the treatment and baseline conditions.

Our study naturally contributes to three important and active bodies of literature in economics. We briefly review these connections---to network public goods games, to structural estimation of network formation data, and to social preferences for trust and reciprocity---in the remainder of this section.

\subsection{Network provision of public goods}

There has been prior work examining both the extension of public goods to static (exogenous) networks, and the provision of public goods on endogenous networks. In particular, \cite{bramoulle2007public} launched research into this environment, by showing that given a network shape, specialized Nash equilibria (in which a small subset of players provide effort, and those connected to them free-ride) tend to be not only stable but also the most efficient in terms of overall welfare. Further, they find that, in an environment with agent heterogeneity, these specialized equilibria are often unique. This is especially relevant since, in our endogenous network environment, an agent's marginal cost of effort is dependent on their position in the network, and in particular how many others they choose to share with. 

\cite{elliott2019network} characterize outcomes in public goods games on exogenous networks by the spectrum of a matrix called the benefits matrix, in which each entry gives the marginal rate of substitution between decreasing own contribution and increased benefits from a neighbor in a fixed network. Their results tie the existence of Pareto-efficient outcomes to the spectral radius of the benefits matrix and characterize the Lindahl outcomes as those with effort proportional to each individual's eigenvector centrality in the graph described by the benefits matrix. Many of their results rely on the connectedness of the benefit graph, which is not guaranteed in random graphs. In the case of endogenous formation this assumption may not be satisfied, rendering spectral methods difficult to implement outside of particular special cases involving links that are fully independent \citep[e.g.][]{dasaratha2020distributions, parise2023graphon}.

Finally, in the exogenous/fixed network case, \cite{boosey2017conditional} uses data from a laboratory experiment to examine the mechanisms for cooperation in a repeated network public goods game. Experimental results showed a significant portion of subjects playing strategies of \textit{conditional cooperation}, in which subjects play strategies which react strongly to the behavior of their neighbors in previous rounds. We incorporate this phenomenon into our structural model, by placing strategies on an evolutionary spectrum from reactionary to predictive. When playing a purely reactive strategy under bounded rationality, simultaneous play may not converge to a stage-game equilibrium \citep{alos2010logit, hommes2012multiple}. 

Prior extensions of public goods provision to environments with endogenous linking include \cite{galeotti2010law}, which furthers the specialization result of \cite{bramoulle2007public}. These papers emphasize the prevalence of core-periphery architectures as equilibrium networks, but in a setting where players choose incoming, rather than outgoing links. \cite{kinateder2017public} expand on this result, showing how heterogeneity in costs and benefits to effort in such an environment generates equilibrium networks which take the form of a tiered multipartite graph. In the heterogeneous case, individuals with comparative advantage in production costs bear most of the brunt of contribution, while peripheral players, corresponding to those with higher valuations of the good, pay the linking costs associated with linking to the core (or inner levels of the multipartite network). 

In another relevant study by \cite{rand2011dynamic}, the authors conducted an experiment to gauge the effects of endogenous networks on cooperation in a repeated prisoner's dilemma. By varying the opportunity for network updates, they showed that subjects are able to take advantage of their ability to change social ties in order to refine their neighborhoods and increase efficiency. Our results show that, while the endogeneity of the network itself does allow for this fine-tuning of the social neighborhood, the dynamics of the network alone are not sufficient to support long-term efficient outcomes. Instead, a platform that aims to nudge players toward efficient social structure should take advantage of its ability to shape and distribute information to its users.

Our model departs from the existing literature on public goods in endogenous networks in a number of ways. Primarily, we model a situation in which individuals choose others with whom they would like to share the externalities generated by their resources. This is the reverse of the situations studied in the previous literature on public goods and sharing on endogenous networks \citep[e.g.][]{galeotti2010law,kinateder2017public, brown2024team}, in that individuals choose the \textit{outgoing} direction of their externalities, rather than the incoming direction of others' externalities. The cost of linking in this study's environment is explicitly tied to the effort or contribution level and can be flexibly specified to represent pure or impure (congestive) externalities. Also in this voluntary sharing environment, there is a unique stage-game Nash equilibrium of no contributions and no linking. This is, however, in stark contrast to what we actually observe in the laboratory implementation, and provides a rich environment to identify and analyze the structure of social preferences.


\cite{hiller2022simple} presents a similar model of network formation with a continuous effort choice and where effort levels are strategic complements. In their paper, while there are congestion costs associated with linking highly, these congestion costs are framed as fixed costs associated with additional links rather than as an increased marginal cost of effort. In our setting, actions are perfect substitutes and the only complementarities among players' actions are those created by their preferences for reciprocity. Finally, \cite{dasaratha2023innovation} also derives theoretical results for a similar model of collaboration with a specific application to collaborative innovation in industrial networks.

\subsection{Structural estimation from network data}
In addition to the novelty of the sharing environment, we also contribute to the growing literature on estimation in network formation games; a problem which is notoriously difficult to analyze, but with important applications \citep{chandrasekhar2016econometrics}. Thus far, the literature has focused predominantly on preference estimation from large-network cross-sectional data. This is not without reason---panel data consisting of repeated observations from small networks over time is rare and difficult to find in practice. By collecting experimental data, however, we can carefully control group size and monitor progress over time; generating a panel of networks from which estimation is particularly tractable and relies on fewer assumptions. More specifically, by observing a panel of small networks over time, rather than that of a single large cross-section, we are able to estimate the discrete choice process directly, reducing the set of computational issues typically encountered in large-network cross-sectional estimation. Estimation of a network formation game as a combination of individual evolutionary discrete choice strategies is a new approach, examined only very recently by \cite{overgoor2019choosing} and \cite{gupta2022mixed}.

Cross-sectional network formation estimators rely on assumptions about the meeting process and dynamics that guarantee convergence to a stochastically stable stationary distribution, also called a Quantal Response Equilibrium (QRE). While the QRE \citep{mckelvey1995quantal,mckelvey1998quantal} is a fixed point stationary distribution of the logit-response (logit best-reply) dynamics, in the case of simultaneous revision opportunities this fixed point is potentially unstable. This means that play may instead exhibit a Hopf bifurcation and converge to a limit cycle or stable orbit, rather than to the fixed point QRE distribution \citep{alos2010logit}. Estimating the individual strategies, however, rather than imposing stability and estimating the QRE, allows us to comment on the convergence of the calibrated system, as well as to draw from the steady-state distribution under arbitrary specifications of the revision opportunity (that is, since we estimate the utility parameters directly). This means that we can simulate draws of the steady-state distribution for large networks under sequential-move individual revisions, using the methods of \cite{badev2021nash}---the first paper to examine identification in discrete-choice games taking place on endogenous networks---in which agents choose both a set of links and an action or investment level. The method used is closely related to the one used by \cite{mele2017astructural} to estimate structural parameters in such a setting, but differs by leveraging panel structure to avoid the assumption that the network has already converged to its steady-state distribution after sufficient iterations of an individual revision process.

\subsection{Social preferences for trust and reciprocity} There is a large literature in behavioral and experimental economics that points toward the importance of various behavioral traits and heterogeneous characteristics of trust and reciprocity in sharing behavior. A large series of studies including \cite{fehr1997reciprocity, fehr1998reciprocity, fehr2000fairness, camerer2003behavioural, cox2004identify} have leveraged experimental evidence to highlight the use of trust and reciprocity as devices for contract enforcement and for driving cooperation in markets and sharing games. Throughout these works it is emphasized that reciprocity manifests not only positively, but can also be used to characterize punishment behavior. While trust and reciprocity in these settings often drive efficiency gains, we show through counterfactuals that which characteristics matter the most is highly context-dependent; both trust and reciprocity can backfire as tools to promote efficiency, depending on the information structure.

This result agrees with some more recent work examining the role of these characteristics in supporting positive market outcomes. For example, \cite{choi2022market} finds evidence suggesting that providing reputation systems in experimental markets interacts with preferences primarily by giving participants more information about whom not to trust. Subsequent work by \cite{solimine2023reputation} supports this result and further emphasizes the role of the information in determining the effectiveness of trust in promoting positive market outcomes. Our counterfactual findings involving trust agree with these findings; through the way that trust interacts with preferences for reciprocity and altruism, promoting trust in the community dramatically improves outcomes when subjects are provided with detailed information about others' behavior. When information is more limited, however, introducing higher levels of trust increases trustworthy behavior by some but may backfire by allowing others to take advantage of this change.

These results also agree with other experimental work that has specifically focused on the role of trust in experimental sharing environments. In particular, \cite{glaeser2000measuring} finds that survey questions about trust (similar to the questions we used to measure trust) are effective at predicting trustworthy behavior but less so at predicting trusting behavior. The distinction between trusting and trustworthy behavior is mirrored in \cite{anderson2004social}, who find that certain measures of trust are negatively associated with sharing in public goods experiments.

Studies that focus specifically on reciprocity have found similar results. \cite{fehr2000fairness} explains that ``[reciprocity] means that in response to
friendly actions, people are frequently much nicer and much more cooperative than predicted by the self-interest model; conversely, in response to hostile actions they are frequently much more nasty and even brutal.'' Indeed, reciprocity is not always positive and evidence from the intersection of neuroscience and economics \citep{fehr2005neuroeconomic} agrees that reciprocity can really be broken down into two pieces---rewards from mutual cooperation and a taste for punishment of unfair behavior. Through a factor analysis of survey responses, we break reciprocity into these two components; one that measures ``overall reciprocity'' and positively weights both types, and one component that distinguishes between positive and negative reciprocity. We find that while increasing trust can mildly increase free-riding in the absence of information, and promoting overall reciprocity can backfire in the presence of information by stimulating punishment. However, an intervention that is specifically targeted at moving players away from punishment and instead focuses on rewarding mutual cooperation can be effective at promoting collaboration in both information environments.

\subsection{Outline of the paper.} In Section~\ref{sec:theory}, we lay out a simple theoretical framework for the collaborative sharing environment. Section~\ref{sec:experiment} describes the design and procedures for the laboratory experiment testing the effects of different information structure on collaboration patterns. We present and discuss the reduced-form results of the experiment in Section~\ref{sec:results}. Then, in Section~\ref{sec:structural} we develop and estimate a structural empirical framework to analyze social preferences using our experimental data, and discuss the results of the estimation and the goodness of fit. This section also presents the three counterfactual simulations designed to measure the value of social preferences for trust and reciprocity. Section \ref{sec:conclusion} concludes.

\section{Conceptual Framework \label{sec:theory}}

We represent player decisions in a single period by directed networks with implicit self-links, coupled with a vector of actions. Consider a set of players (agents) $N= \{ 1 \dots n \}$. An individual agent in this set is denoted as $i$. Players have a resource constraint of $\omega_i$ units, and simultaneously choose a contribution level $c_i$ from the interval $[0,\omega_i]$. Thus $(\omega_i - c_i)$ are kept at no cost, and used to produce one unit of value for player $i$.

In addition to choosing a contribution level, players choose a subset $N_i$ of the other players, with whom they would like to split the production cost and benefits. The cost of contribution and the magnitude of externalities are flexibly specified by a \textit{Marginal Per Capita Return} (MPCR) function $m_i(N_i)$, which maps the player's chosen neighborhood to a cost-externality structure. In this way, the marginal cost of contributing is $(1-m_i(N_i))$. For example, in the case of homogeneity in costs and congestive externalities, which will be a focal point in this study, we could set $m_i(N_i) = m(N_i) = \tfrac{1.6}{|N_i|}$ for $|N_i| > 0$. Without loss of generality, we fix $m_i(\emptyset) = 1$. Because we are primarily interested in examining behavior in situations where sharing is not apparently rational under standard, self-interested preferences, we operate under the following assumption:
\begin{assumption}$m_i(N_i) = 1$ for the empty neighborhood $N_i=\emptyset$, and $m_i(N_i) < 1$ for all nonempty neighborhoods $N_i \neq \emptyset$.\label{asn:boundedmpcr} \end{assumption}

This assumption guarantees that players cannot generate efficiency gains without sharing with another player. The outcome of the game in any round can be summarized by an $n\times n$ \textit{adjacency matrix} $A$, with the elements $A_{ij} = 1$ if player $j$ shares with player $i$, and $A_{ij} = 0$ otherwise. Since the benefits of investment are divided between the contributor and their neighborhood, we will fix $A_{ii} = 1$. We will denote $A_i$ as the $i^{th}$ column of $A$ with the $i^{th}$ element removed. Thus $A_i$ is an $(n-1)$-vector from $\{0,1\}^{(n-1)}$ that contains all the information of the neighborhood $N_i$ of a node $i$. In addition, we collect the contributions $c_i$ into a single vector $c$. By slight abuse of notation, we will overload the marginal per-capita return function $m_i(\cdot)$ to accept arguments in the form $A_i$, with $m_i(A_i) \triangleq m_i(\sum_{j=1}^{n-1} (A_i)_j)= m_i(|N_i|)$.

Then we can write players' concrete (observable) monetary payoffs in the following form:\begin{equation}
   \pi_i (A,c) = \pi_i (A_i,c_i ; A_{-i}, c_{-i}) = \left( \sum_{j=1}^n A_{ij}  c_j m_j \left(A_j\right) \right) - c_i
\end{equation}

Importantly, the summation includes $j=i$ because $i$ shares in the benefits of their own investment, although in our cases of interest, they cannot earn positive net return on their own contribution. Incentives can be decomposed into a combination of effort cost and externalities from others' contributions. That is, we can write $\pi_i$ as: \begin{equation}
    \pi_i (A,c) = (m_i (A_i) - 1) c_i + \kappa_i (A_{-i}, c_{-i} )
\end{equation}
where $\kappa_i (A_{-i},c_{-i}) = \sum_{j=1, j\neq i}^n A_{ij} c_j m_j(A_j)$ models pure externalities, which do not affect the decision problem of agent $i$ on the margin, since they cannot be controlled. It is straightforward to show, by backward induction (and using Assumption~\ref{asn:boundedmpcr}), that the following proposition holds:\begin{proposition}
The unique Nash equilibrium of the sharing game (under Assumption~\ref{asn:boundedmpcr}) is $c_i = 0$ and $N_i = \emptyset$.
\end{proposition} \begin{proof}
Since $\kappa(A_{-i},c_{-i})$ is independent of both $c_i$ and $A_i$, this term is dropped from the marginal decision problem when taking first-order conditions. We have:\begin{equation}
    \frac{\partial\pi_i}{\partial c_i} = m_i(A_i) - 1 < 0
\end{equation}
by Assumption \ref{asn:boundedmpcr}. Thus in the stage game, the optimal contribution for any linking pattern is $c_i = 0$ for all $i$.
\end{proof}
The free-riding hypothesis is a prominently studied feature of the voluntary contributions mechanism. Robust evidence from both the lab and field, however, shows that while this \textit{is} the unique equilibrium, realistic play very rarely agrees with this theoretical result of full free-riding \citep{isaac1988group,fisher1995heterogenous}.

Explanations for the failure of the free-riding hypothesis typically fall into two categories---bounded rationality and behavioral or social preferences. Bounded rationality models assert that players make errors in their computation of the optimal strategy. Social preferences models, on the other hand, aim to explain systematic deviations by asserting that the strategies are, in fact, rational and optimal when considering other non-monetary concerns of the players---such as preference for reciprocity or fairness concerns---which bias actions toward prosocial behavior. We describe the mean utility as allowing flexibly specified social preferences:\begin{equation}
    u_{i} (A,c\vert \theta) = \theta_1 \pi_i (A,c) + \beta_i (A,c\vert \theta)
\end{equation}
where $\beta_i$ is a general function specifying the social preferences. Since the externality term $\kappa$ is entirely independent of the action choice of player $i$ in period $t$, we can aid computation by truncating utility to form \textit{individual potentials}:\begin{equation}
    \phi_{i} (A,c\vert \theta) = \theta_1 \left(\pi_i(A,c)- \kappa_i (A_{-i},c_{-i})\right) + \beta_i (A,c | \theta) 
\end{equation}
An effective alternative solution concept should incorporate behavioral preferences, taking $\theta$, which is a vector of parameters, as given. We present the following stability concept referred to as behavioral Nash stability:\begin{definition}\label{def:behavstable}
A state of the network and contribution profiles is \textit{behaviorally Nash stable} if the following conditions hold for all players $i$:\begin{enumerate}
    \item$\phi_i(A_{i}, c_{i} | A_{-i}, c_{-i},\theta) \geq \phi_i(A_{i}', c_i | A_{-i}, c_{-i},\theta) \quad \forall A_i' \in \{0,1\}^{n-1}$\label{cond:1}
    \item One of the following:\begin{itemize}\item$\frac{\partial}{\partial c_i} \beta_i(A,c\vert\theta) = \theta_1(1-m_i(A_i)) $\label{cond:2}
    \item$c_i=\omega_i \;and\; \frac{\partial}{\partial c_i} \phi_i (A,c\vert \theta) > 0$\label{cond:2a}
    \item $c_i = 0\; and\; \frac{\partial}{\partial c_i} \phi_i(A,c\vert\theta) < 0$\label{cond:2b}\end{itemize}
\end{enumerate}
\end{definition}
This equilibrium concept is similar to those used in prior literature studying games on endogenous networks \citep{golub2021games} and, in particular, the notion of ``k-stability'' introduced by \cite{badev2021nash}. Since players can simultaneously and unilaterally update any subset of their links, it is more restrictive than the classical notion of ``pairwise stability" \citep{jackson1996strategic}. The conditions in Definition \ref{def:behavstable} are straightforward from a game theoretic perspective. Condition~\ref{cond:1} guarantees that no player could unilaterally gain by changing their link strategy. Condition~\ref{cond:2} stipulates that the contribution profile must be a Nash equilibrium of the game induced by the behavioral preferences and the network $A$. 

It is also straightforward to characterize what is meant by an efficient outcome. When discussing efficiency, we refer only to the direct monetary outcomes of gameplay.\begin{definition}
    An outcome $(A,c)$ is \textit{efficient} if $\sum_{i=1}^n \pi_i(A,c) \geq \sum_{i=1}^n \pi_i(A',c')$ for all $(A', c') \in \mathcal{A}\times\mathcal{C}$
\end{definition}
In other words, an outcome is efficient if it generates the maximum possible monetary gains for the players. Finally, consider the following classification of congestion effects.
\begin{definition}
    The decision setting is called\begin{itemize}
    \item \textit{Purely congestive} if $m_i(|N_i|) = \frac{k}{|N_i|+1}$ for $|N_i|>0$ and $k\in \mathbb{R}$ with $1 < k < 2$
    \item \textit{Subcongestive} if $m_i(|N_i|) = \frac{k(|N_i|)}{|N_i|+1}$ for $|N_i| > 0$ and $k:\mathbb{Z}_{n}\setminus \{0\} \longrightarrow \mathbb{R}$ is a decreasing function with $1 < k(1) < 2$
    \item \textit{Supercongestive} if $m_i(|N_i|) = \frac{k(|N_i|)}{|N_i|+1}$ for $|N_i| > 0$ and $k:\mathbb{Z}_{n}\setminus \{0\} \longrightarrow \mathbb{R}$ is an increasing function with $1 < k(m) < m$ for all $m\in \mathbb{Z}_n \setminus \{0\}$
    \end{itemize}
\end{definition}
This characterization is inspired by the literature on group size effects in voluntary contributions environments \citep[e.g.][]{isaac1988group}. Pure congestion describes a situation in which players generate fixed efficiency gains by sharing, and gains from their investment are divided between the sharing player and their neighbors of choice. Subcongestivity and supercongestivity refer to settings in which the marginal per-capita return decreases faster or slower with the addition of a link than the purely congestive return. Supercongestivity also contains an interesting special case, which might be termed \textit{anticongestive}, in which the good exhibits a specific type of ``network effect'' -- that is, sharing becomes more effective as the group size increases.

This special case lends itself to a convenient characterization of efficient structure, based on how a player's MPCR scales with their out-degree. That is, the efficient outcome depends on the curvature of the MPCR function. It transitions between these phases when the game is purely congestive, at the boundary between subcongestivity and supercongestivity. Along this boundary, any network in which every player is participating can support efficient outcomes. \begin{proposition}\label{prop:efficiency} A network is efficient in a:\begin{itemize}
    \item Purely congestive game if and only if $|N_i| \geq 1$ and $c_i = \omega_i$ for all players $i$.
    \item Subcongestive game if and only if $|N_i| = 1$ and $c_i = \omega_i$ for all players $i$.
    \item Supercongestive game if and only if $|N_i|=n$ and $c_i = \omega_i$ for all players $i$.
\end{itemize}
\end{proposition}\begin{proof}
 In a purely congestive setting, efficiency gains are characterized by the multiplicative constant $k$. More specifically, since each player's contribution is multiplied by $k$ and divided equally among themselves and their neighborhood, we have \begin{equation}\sum_{i=1}^n \pi_i(A,c) = \sum_{i=1}^n \left( \frac{k}{|N_i|+1}c_i + \sum_{j=1}^n A_{ji} \frac{k}{|N_j|+1}c_j\right)= \sum_{i=1}^n kc_i = k\sum_{i=1}^n c_i,\end{equation} as long as each player forms at least one link. This efficiency is linear and strictly increasing in $c_i$ for all $i$ because the definition of pure congestion prescribes $k>1$. Therefore the sum of all monetary payoffs is strictly increasing in contributions, meaning that an efficient configuration must be characterized by full contribution.
 
 When the setting is sub- or super-congestive, $\sum_{i=1}^n \pi_i (A,c) = \sum_{i=1}^n k(|N_i|) c_i$. Since both $k(\cdot)$ and $c_i$ are restricted to be positive, the function is linear and increasing in both terms. Thus, efficiency prescribes that the players choose the network that generates the largest value of $k(\cdot)$. If $k(\cdot)$ is monotone, this means choosing either the smallest (if the game is subcongestive) or largest (if the game is supercongestive) possible number of links. Finally, Assumption \ref{asn:boundedmpcr} guarantees that $k(0) = 1$ (and thus that $m(0)=m(\emptyset) = 1$. Since $k(\cdot) > 1$ for nonzero arguments in either subcongestive or supercongestive games, any fully efficient structure in these settings can never include an isolated node.
\end{proof}

For the experiment, we turn our focus to the purely congestive case, using the number of free-riders in a group to describe an efficient structure.

     \section{Experimental Design \label{sec:experiment}}
     
     In this section, we describe the design and procedures of the laboratory experiment in greater detail. The experiment was conducted using undergraduate students in the XS/FS Experimental Social Sciences Laboratory at Florida State University. We collected data from a total of 184 subjects across eight sessions. Subjects were recruited using ORSEE \citep{greiner2015subject}, and played a computerized version of the game programmed using zTree \citep{fischbacher2007z}. Instructions used in the experiment, including screenshots of the decision screens, are contained in Appendix~\ref{app:instructions}.
     
     Subjects played a repeated version of a purely congestive resource sharing game, which consisted of a fixed group size of $n=4$, and a homogenous cost/externality structure defined by an MPCR function of $m_i(N_i) = m(N_i) =\frac{1.6}{|N_i|}$. Subjects were informed at the beginning of the session that there would be two parts to the experiment but were not informed about any details of the second part until the first part was completed. Each subject in a group was assigned a unique ID (1--4). In the first part, which consisted of 15 rounds, no information was shown to subjects regarding the individual decisions made by others in their group. Instead, subjects were simply shown their own payoff between rounds. At the end of round 15, subjects were redirected to a waiting screen, at which point the instructions for the second part were read.
     
     The treatment variation was implemented in the second part. In three baseline sessions, consisting of a total of 72 subjects in 18 groups, subjects were told that the second part of the experiment would be exactly the same as the first part, except that subject IDs would be randomly reassigned. In five treatment sessions, consisting of 112 subjects across 28 groups, subjects were also told that they would play the game for another 15 rounds. However, in addition to reassigning ID's, subjects were also told that they would be shown how much benefit they received in the previous round from each other subject in their group. Although providing subjects with information about the past behavior of others in their group does not change the unique Nash equilibrium, this information treatment facilitates direct reciprocity where the baseline sessions do not. After the two main parts of the experiment were finished, subjects completed a series of questionnaires designed to elicit behavioral characteristics. Questions from this section are shown in Appendix~\ref{app:questions}. Sessions lasted no longer than an hour. At the end of the session, subjects were paid privately by check, earning an average of \$16.69, including a \$10 show-up fee.
     
     While providing subjects information about the past behavior of others in their group does not change the unique Nash equilibrium, we expect that it will have a nontrivial impact on efficiency and performance in the market. As such, we estimate the effects of the treatment on a number of outcome measures, and use structural methods to investigate the root causes of these changes in behavior.
     
    \section{Reduced-Form Results \label{sec:results}}
    In this section, we highlight several aggregate, reduced-form results showing the treatment effects of information provision on key outcomes of interest. The reduced-form estimates are presented in Table~\ref{res:reduced}. Moreover, Figure~\ref{fig:did} illustrates the evolution of the mean level (for each key outcome) across experimental periods, with bars representing standard errors clustered at the group level. All estimations highlight a strong impact of the information treatment on key outcomes. For computational convenience, in each estimation, we normalize contributions to lie in the interval between $0$ and $1$.
    

\begin{figure}
\centering
\subcaptionbox{Contributions (in Experimental Currency Units)\label{sfig:contributions}}{\includegraphics[width=0.49\linewidth]{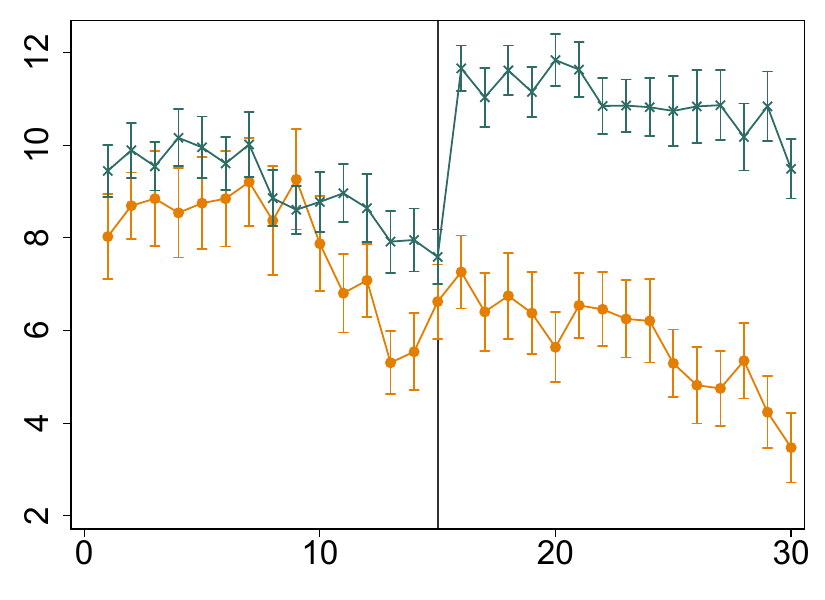}}
\subcaptionbox{Links\label{sfig:links}}{\includegraphics[width=0.49\linewidth]{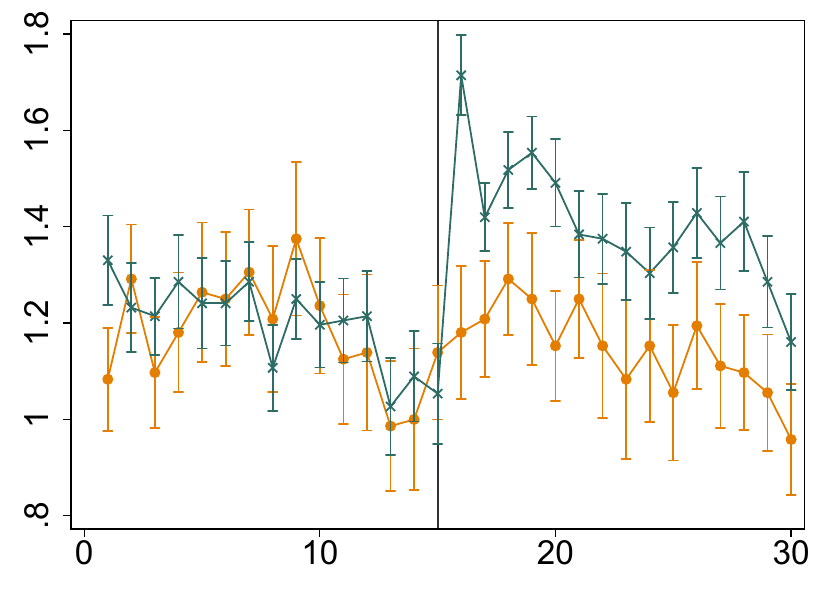}}\\[1.5em]
\subcaptionbox{Isolated Nodes\label{sfig:isol}}{\includegraphics[width=0.49\linewidth]{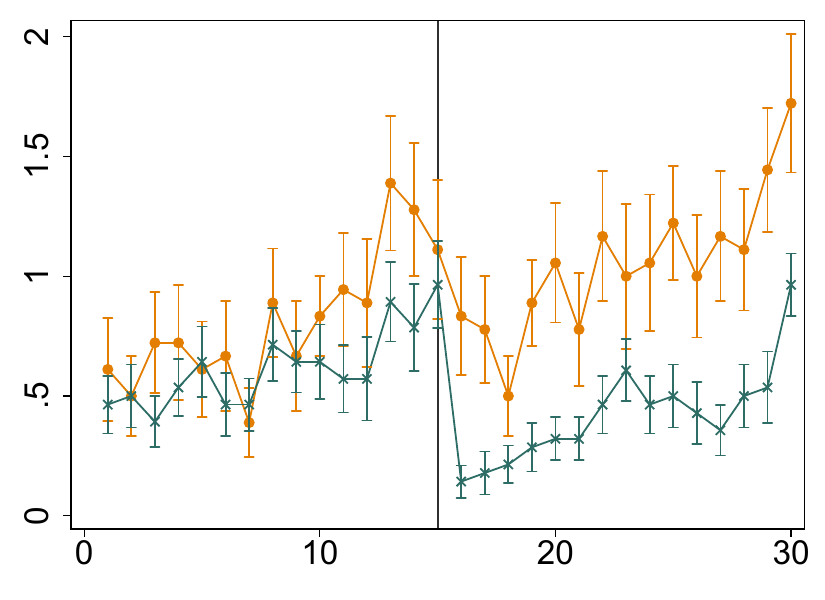}}
\subcaptionbox{Individual Costs\label{sfig:cost}}{\includegraphics[width=0.49\linewidth]{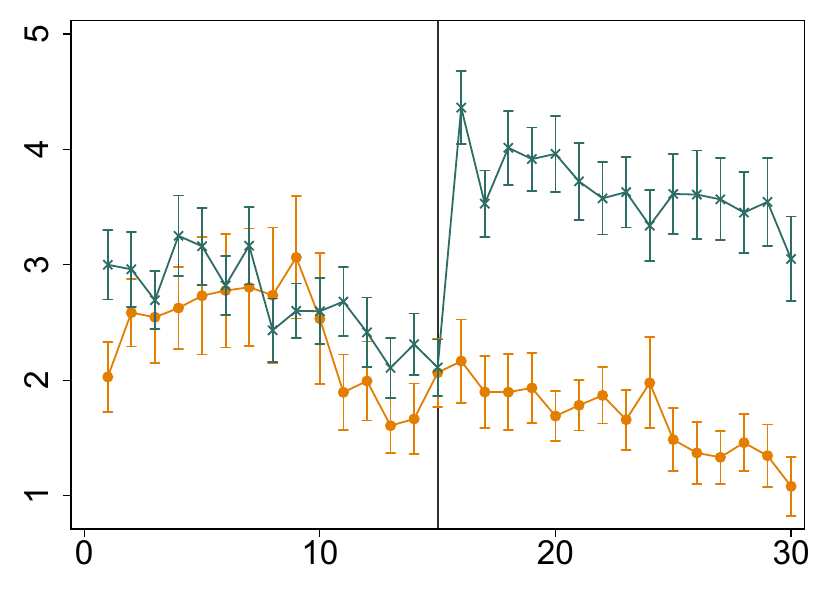}}\\[1.5em]
\subcaptionbox{Reciprocity\label{sfig:rec}}{\includegraphics[width=0.49\linewidth]{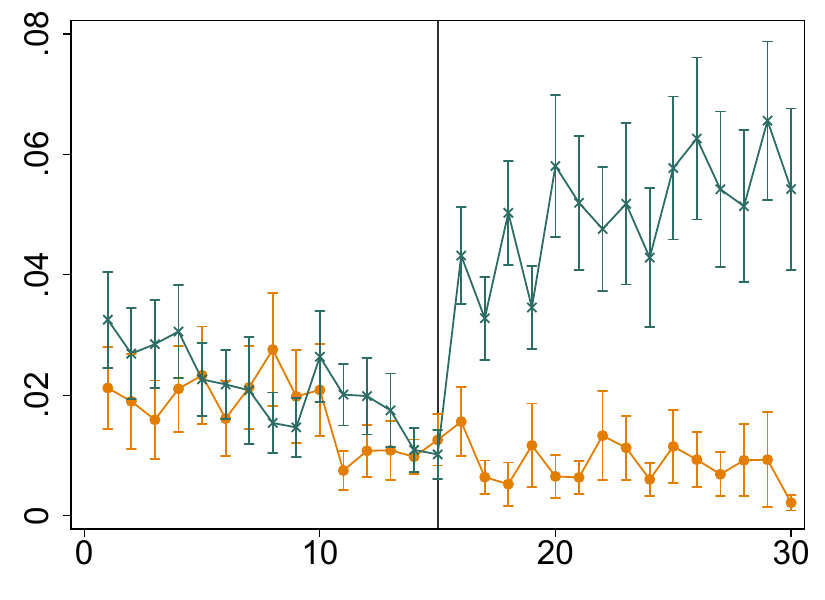}}
\subcaptionbox{Centralization\label{sfig:cent}}{\includegraphics[width=0.49\linewidth]{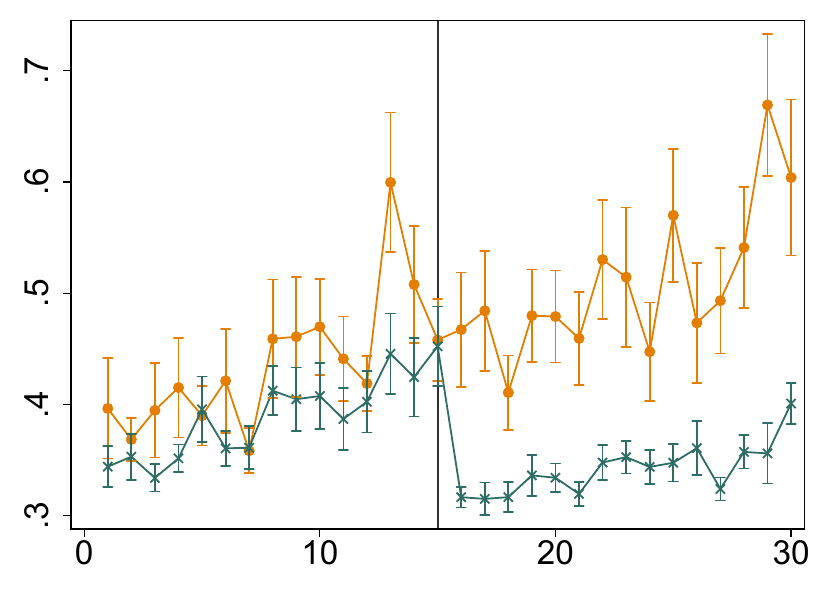}}\\[1.5em]
\subcaptionbox*{}{\includegraphics[width=0.4\linewidth]{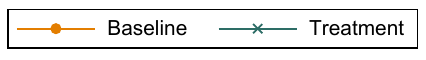}}\\[-2em]
\caption{\textbf{Dynamics of key outcomes}}\label{fig:did}
\end{figure}



\subsection{Contributions and links} A natural question is how contributions and average degree (number of outgoing links) are impacted by the information treatment. The dynamics of these variables can be found in Figures~\ref{sfig:contributions} and \ref{sfig:links} along with corresponding estimation results in columns (1) and (2) in Table~\ref{res:reduced}. We find that the treatment substantially increases both contributions and linking. Contributions still show a tendency to decrease over time, and the rate of this decay is not significantly impacted by the treatment. Links, on the other hand, exhibit an additional differential dynamic. Although there is a strong initial boost from the treatment intervention, decay in the number of links actually appears to accelerate mildly relative to groups in the baseline sessions.

\begin{result}
    Contributions and average degree (number of links) are both significantly higher in the treatment condition than in the baseline condition.
\end{result}

\begin{table}[t]\small\caption{\textbf{Treatment effects on key outcomes}\label{res:reduced}}
   {
\def\sym#1{\ifmmode^{#1}\else\(^{#1}\)\fi}\scriptsize
\begin{tabular}{l*{6}{c}}
\hline\\
            &\multicolumn{1}{c}{(1)}&\multicolumn{1}{c}{(2)}&\multicolumn{1}{c}{(3)}&\multicolumn{1}{c}{(4)}&\multicolumn{1}{c}{(5)}&\multicolumn{1}{c}{(6)}\\
            &\multicolumn{1}{c}{Contributions}&\multicolumn{1}{c}{Links}&\multicolumn{1}{c}{Efficient structure}&\multicolumn{1}{c}{Costs}&\multicolumn{1}{c}{Reciprocity}&\multicolumn{1}{c}{Centralization}\\
\hline\\
Period           &     -0.0319&     -0.0253  &     -0.0105&      -0.0106&   -0.000667&     0.00640\\
                    &   (0.00455)         &    (0.0118) &   (0.00295)         &   (0.00212)         &  (0.000160)         &   (0.00100)         \\
[1em]
Treatment         &       0.783&       1.668 &       0.375 &      0.353&      0.0252&      -0.124\\
                    &    (0.0900)    &     (0.252) &    (0.0593)          &    (0.0446)         &   (0.00482)         &    (0.0173)         \\
[1em]
T $\times$ Period        &      0.0108         &     -0.0644 &     -0.0147  & 7.58$\times$10$^{-6}$         &     0.00202  &   -0.00243         \\
                  &    (0.0101)         &    (0.0274) &   (0.00436)      &   (0.00577)         &  (0.000805)         &   (0.00150)         \\
[1em]
Group FE      &       Yes&      Yes&       Yes&     Yes&      Yes&      Yes\\       \\
\hline
\(N\)       &        1380         &        1380         &        1380         &        1380         &        1380         &        1362         \\
\hline
\multicolumn{7}{l}{\footnotesize Standard errors in parentheses, clustered for 46 groups}\\
\end{tabular}
}

\end{table}

\subsection{Efficient Structure} The estimation in column (3) of Table~\ref{res:reduced} shows the effects of the treatment on subjects' abilities to coordinate on efficient structure. By efficient structure, we refer to a network topology that satisfies the conditions required for efficiency by Proposition~\ref{prop:efficiency}, without necessarily satisfying the requirement of full contributions. In other words, this is a binary indicator taking the value $1$ if the observed network structure is theoretically capable of supporting efficient outcomes, and takes the value $0$ otherwise. Because the version of the game we implemented was purely congestive, an efficient structure is one in which there are no ``isolated'' players who do not form any outgoing links. Figure~\ref{sfig:isol} shows how the number of isolated players changes over time. As shown in the figure and estimates, the treatment has an immediate impact in supporting efficient structure. This is evidenced by a large downward jump in the number of isolated players at the time of intervention---and correspondingly, a significant positive effect of Treatment on Efficient Structure in column (3). The likelihood of an isolated node, however, continues to increase over time after the intervention; and similar to the average degree (number of links), its growth appears to accelerate. This indicates that degree reductions are not coming exclusively from highly connected players, which in turn impacts the structural ability of the social system to reach efficient outcomes.

\begin{result}
    The information treatment has a significant positive impact on the formation of efficient network structure, by reducing the prevalence of isolated nodes.
\end{result}

\subsection{Individual Costs} Contribution decisions and link decisions in isolation may not capture the true dynamic of behavior. This is because both variables together determine the cost of sharing---the marginal return on a player's contribution in this purely congestive game decreases as they share with more other players. For each observed action, we compute the direct cost to the sharing individual as $(1-m_i( A_i))c_i$. The dynamics of these individual costs are illustrated in Figure~\ref{sfig:cost}. Due to the decreasing trend in both contributions and linking, costs follow a downward trend. Treatment sessions see a large fixed jump in individual costs, and like contributions, the differential effect is not significantly different from zero, indicating that costs continue to decrease at the same rate as they did before the intervention, albeit from a significantly higher starting point.

\begin{result}
    Individual costs of sharing are significantly higher in the treatment condition than in the baseline, although they exhibit a similar downward trend over time in both conditions.
\end{result}

\subsection{Reciprocity} A commonly cited behavioral driver of sharing behavior in similar environments, and a natural candidate for behavioral preferences, is reciprocity \citep[e.g.][]{fehr1997reciprocity,fehr1998reciprocity}. Our main hypothesis is that individuals in the treatment condition can use their new information to coordinate with other players who are sharing with them. This type of reciprocity would make contribution decisions locally complementary. Thus, a natural measure for such reciprocity is the product of incoming and outgoing benefit flows from each player. We compute this measure for each network cross-section as $\sum_{i=1}^n \sum_{j=1,j\neq i}^n A_{ij}A _{ji}m_i(A_i)m_j(A_j) c_i c_j$. This measure exhibits a negative trend in the baseline sessions as players fail to successfully coordinate on reciprocal outcomes in the absence of information. 

The information intervention has a positive effect both immediately and differentially. In the treatment condition, reciprocity displays a large positive level effect and switches from following a downward to an upward trend. These trends are apparent in Figure~\ref{sfig:rec} and supported by the coefficient estimates in column (5) of Table~\ref{res:reduced}. The upward trend in reciprocity after the treatment intervention indicates that the downward trending individual costs in treatment sessions can be explained primarily by players progressively pruning unreciprocated links. It also suggests that, even using information only from the past, players can succeed in coordinating on stable concurrently reciprocated links.

\begin{result}
    Reciprocity is significantly higher in the treatment condition than in the baseline condition. Moreover, it increases over time in the treatment, while it decreases over time in the baseline.
\end{result}

\subsection{Centralization} Finally in Figure~\ref{sfig:cent} and column (6) of Table~\ref{res:reduced}, we estimate the impact of the treatment on (de)centralization. To measure centralization, we use the Herfindahl-Hirschman index (HHI), which is computed as $\sum_{i=1}^n \left(\frac{1.6*c_i}{\sum_{i=1}^n 1.6*c_i}\right)^2$ (restricting focus to situations in which there is at least one player who is sharing). We observe that centralization tends to increase over time, as a small set of players tend to emerge who end up generating most of the efficiency gains for the group. The treatment has an immediate impact in encouraging decentralized networks, characterized by a more balanced profile of actively sharing players. Point estimates also suggest a differential impact of the treatment, in that it tends to not increase as quickly in treatment sessions as in the baseline. However, our estimate of this effect is not quite significant at the $10\%$ level, with a p-value of $0.112$.

\begin{result}
    In the treatment condition, we observe more decentralized networks---exhibiting a more balanced profile of contributors---than in the baseline condition.
\end{result}

Taken together, these reduced-form results suggest that subjects are highly effective at using the specific information provided in the treatment condition to coordinate on more efficient collaboration networks. To further investigate the driving forces behind these treatment effects, we adopt a structural approach. In the next section, we introduce our structural framework and implement panel structural estimators of network formation in order to better understand the underlying behavioral traits that generate our reduced-form results.

    \section{Structural Estimation \label{sec:structural}}

\subsection{Structural Model} We use network formation econometrics to structurally model the environment and investigate the root causes of the behavioral changes observed between information environments. When incorporating both prosociality and unobserved decision factors, players make decisions based on the following evolutionary random utility model:\begin{equation}
   \mathbb{E} \left[ U_{i} (A_t, c_t|\theta) | \Omega_{it}\right] = \theta_1\left(\pi_i (A_{t},c_{t}) - \kappa_i (A_{-i},c_{-i})\right) + \mathbb{E}\left[ \beta_i (A,c|\theta) | \Omega_{it} \right] + \epsilon_{it}
\end{equation}
where $\Omega_{it}$ is used to denote the information available to player $i$ at time $t$. Here, utility is linearly decomposed into three separate forces---a weighting of preference for monetary earnings $\pi_i$, a behavioral term $\beta_i$ describing social preferences, and a structural error term $\epsilon_{it}$ which is i.i.d. across link structure alternatives and time. The structural shock $\epsilon_{it}$ represents alternative decision factors, observed by the decision maker but unobserved by the econometrician.

Each player then chooses their strategy to maximize this expected utility, so that:
\begin{equation}
    \mathbb{E}_i \left[ U_i(A_i,c_i | \theta) | \Omega_{it} \right] \geq \mathbb{E}_i \left[ U_i(B_i,d_i|\theta) | \Omega_{it}\right] \quad\forall (B_i, d_i) \in \{ 0,1 \}^{n-1} \times [0,\omega] 
\end{equation}

 Using a model based on beliefs allows us to estimate individual strategic evolution while avoiding the bias that arises as a result of the simultaneous nature of decision-making in this setting. This is enabled by our panel of small networks---the rare but ideal setting in which to study network formation.\footnote{This is as opposed to the potential games approach which uses a single large cross-section and estimates parameters from the stationary distribution, given some assumptions on the symmetry of preferences, and the assumption of nonzero meeting probabilities for all dyads.} Individuals form expectations over the incoming value that will be sent to them from each other player, conditional on the local information that they are given. We denote the information available to the player in period $t$ as $\Omega_t$, representing a set of observations of previous behavior upon which players can form reasonable expectations about their neighbors' future actions. 

We can also relax the usual assumption that players move individually, which is common in the cross-sectional and large network estimators \citep{mele2017astructural,badev2021nash}. This assumption is typically used to ensure the convergence of logit-response dynamics to a steady state distribution \citep{foster1990stochastic,alos2010logit} from which parameters are then estimated. Observing a panel of linking decisions by a subset of nodes, set in small networks, allows us to directly (and tractably) estimate utility parameters from the evolution of gameplay. While the assumption of players using logit best-response is somewhat strict, it is a substantial relaxation of the assumptions used in cross-sectional estimators. Specifically, we no longer need to make assumptions regarding the meeting process for individual revisions.

At time $t$, some (known) subset of nodes are selected to update their linking strategy. We place no restrictions on who or how many are chosen for revision. The nodes that are selected have the opportunity to change their link selection and contribution. They select the combination link set and contribution that maximize their expected utility, including the preference shock which is subject to the following assumption:\begin{assumption}
The structural preference shock $\epsilon_{it}$ is i.i.d. Extreme Value type 1.
\end{assumption}

This allows us to estimate agents' best responses from variations of the multinomial logit:\begin{equation}
    \label{eqn:logit}f_i(A_i,c_i | \Omega_{it}, \theta) = \frac{exp(\phi_{i}(A_i, c_i | \Omega_{it},\theta))}{\int_{0}^{\omega_i}\sum_{G_i \in \{ 0,1\}^{n-1}} exp(\phi_i (G_i,\delta_i | \Omega_{it},\theta))d\delta}
\end{equation}

Using this logit best-response process, the probability of observing network $A$ and contribution profile $c$ in period $t$ is:\begin{equation}
    f(A,c|\Omega_t, \theta) = \prod_{i=1}^n f_i (A_i,c_i | \Omega_{it},\theta)
\end{equation}

\subsection{Methodology} We observe a panel of $N$ groups over $T$ time periods. We denote the network from group $k$ in period $t$ as $A_t^k$. The resulting log-likelihood function is:\begin{equation}
    \ell(\theta) = \sum_{t=2}^T \sum_{k=1}^N \sum_{i=1}^n \ln\left(f_i(A_{it}^k,c_{it}^k \vert \Omega^k_{it}, \theta) \right) = \sum_{t=2}^T \sum_{k=1}^N \sum_{i=1}^n \phi_i (A_{it}^k,c_{it}^k | \Omega_{it}^k,\theta) - \ln \left(Z_{it}(\theta,\Omega_{it}^k)\right)
\end{equation}
where $Z_{it}(\theta,\Omega_{it}^k)$ is the partition function from the denominator of the density \eqref{eqn:logit},
\begin{equation}
    \label{eqn:partition}Z_{it}(\theta,\Omega_{it}^k) = \int_0^{\omega_i} \sum_{G_i \in {\{0,1\}}^{n-1}} \exp(\phi_i (G_i, \delta_i \vert \Omega_{it}^k, \theta))d\delta.
\end{equation}
The MLE is then defined as the vector of parameters that maximizes this objective,\begin{equation}
    \hat{\theta} = \argmax_{\theta\in \Theta} \ell(\theta).
\end{equation}

 The graphs in Figure~\ref{fig:depgraphs} highlight the computational benefits associated with panel estimators of network formation compared to a pooled cross-sectional estimator. In these graphs, nodes represent random variables; a circle represents a single variable (the contribution of Player $i$), and a square represents a dense subgraph in which all of the variables contained in the labeled set are completely connected. A link indicates statistical dependence of the receiving variable on the linking variable. A link to or from a square node indicates dependence between all elements in the union of the sets, and bidirectional links are highlighted in dark red, while dashed grey links are not bidirectional. A light blue node is observed by the decision-maker and thus can be conditioned on in estimation, while all interconnected dark red nodes must be treated as a single outcome variable. Since each player has $n-1$ potential links in the network, each set $A_{i}^{t}$ contains $n-1$ interdependent random variables. 
 
 Panel estimation improves computation time by conditioning on each player's available information to sever concurrent interdependence between links and contributions across players. This works by essentially performing $n$ simultaneous independent multinomial logit estimations across $n$ variables (each potential link plus a contribution level), as opposed to a single multinomial logit over $n^2$ outcomes, saving a total of $2^{(n^2)}-2^{n}$ operations when the support of $c_i$ is binary. For perspective, in a network of size $n=10$, panel estimation reduces the number of required floating point operations from $1.27\times 10^{30}$ to just 1,024 operations per cross-section.

\begin{figure}[!t]
\centering
\subcaptionbox{Cross-sectional\label{fig:xsection}}{\includegraphics[width=0.5\linewidth]{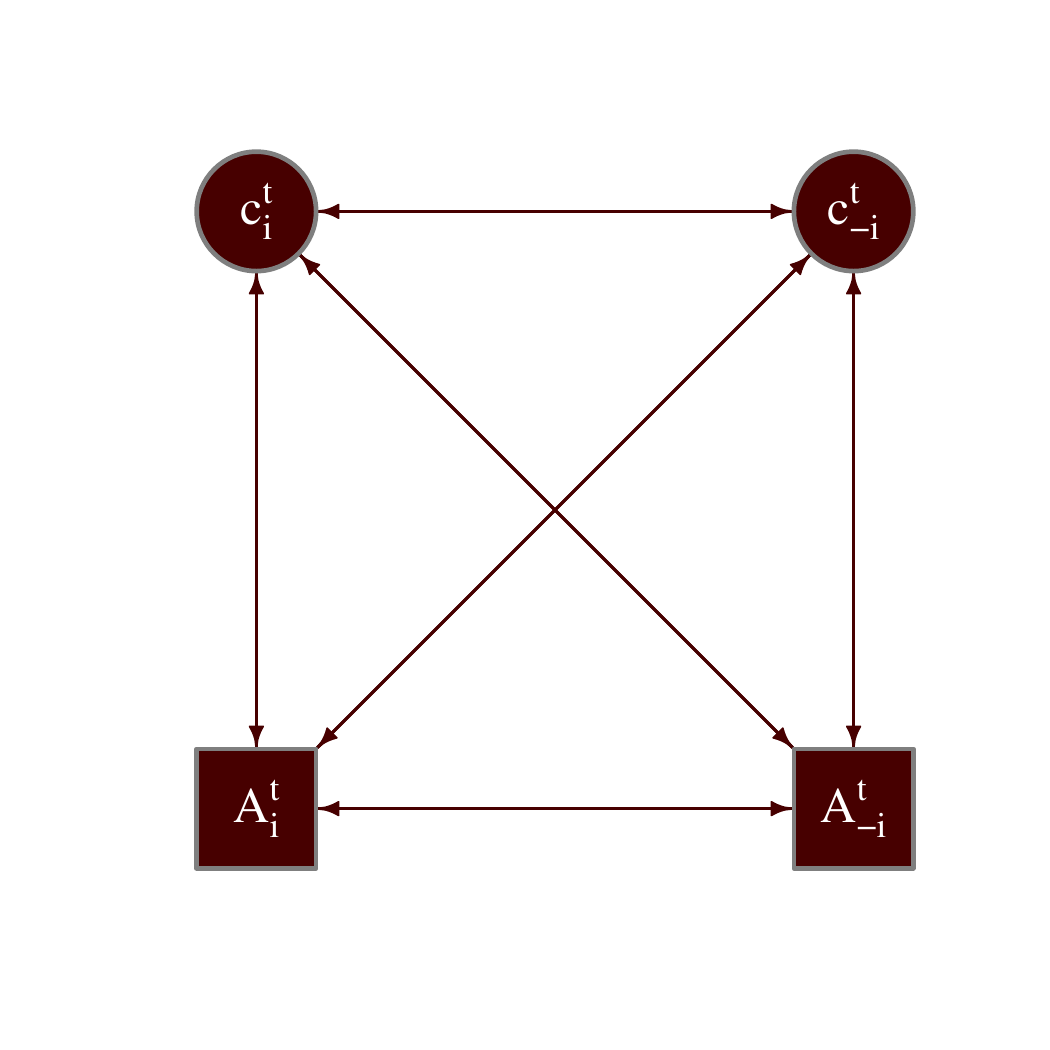}}\subcaptionbox{Panel\label{fig:dg}}{\includegraphics[width=0.5\linewidth]{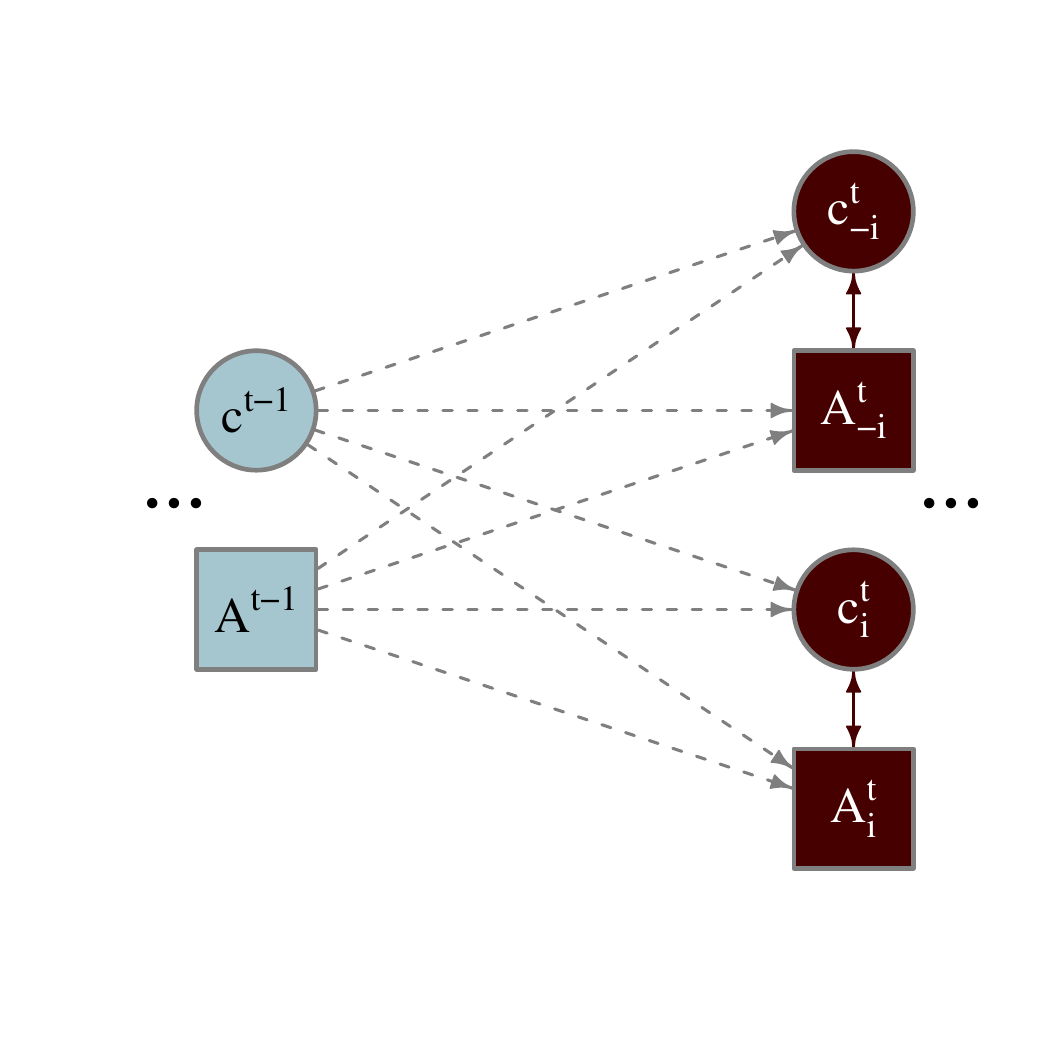}}
\caption{\textbf{Dependence graphs for a two-player game.}}\label{fig:depgraphs}
\end{figure}

\subsection{Inference}\label{ss:inference}

We also estimate uncertainty in these environments by computing asymptotic standard errors using a partial likelihood approach, to account for possible misspecification of the joint density.\footnote{For example, misspecification may arise due to serial correlation in the error term within a group. See chapter 13 of \cite{wooldridge2010econometric} for details of this approach.} The partial likelihood approach assumes only that the conditional distribution is properly specified by $f$, and focuses on a partial likelihood function $\ell$ which is defined as \begin{equation}
    \ell^p_k (\theta) = \sum_{t=1}^T \sum_{i=1}^n \log f_t (A_{it}^k,c^k_{it}\mid \Omega^k_{it}, \theta)
\end{equation} Point estimates from the MLE and PMLE are identical, but the asymptotic variance of the partial likelihood estimator adjusts for clustering by group.

That is, in the limit as $N\rightarrow \infty$, the distribution of $\hat{\theta}_{\text{PMLE}}$ is determined by the central limit theorem\begin{equation}
    \sqrt{N} \left(\hat{\theta}_{\text{PMLE}} - \theta_0\right) \sim \mathcal{N} \left( \mathbf{0}, A_0\inv B_0 A_0\inv \right)
\end{equation}
where $A_0 = -\mathbb{E}\left[ \nabla^2_\theta \ell_k^p(\theta_0) \right]$ and $B_0 = \mathbb{E}\left[ \nabla_\theta \ell_k^p (\theta_0) \left( \nabla_\theta \ell_k^p (\theta_0) \right)^\top \right]$. The key difference between the PMLE and MLE standard error formulas lies in the summation over $t$ within the function $\ell_k^p$, which must be accounted for in the computation of the expected score and Hessian. In this exponential family, the scores can be expressed as the sufficient statistics minus their expected values under $\theta$. Typically, in a large network, this would need to be estimated using Monte-Carlo methods, but conveniently, because the limited group size allows us to directly compute the normalizing constant for each $f_{it}$, we can efficiently compute the scores directly during the calculation of the normalizing constant. Crucially, the summation across time periods must take place before the outer product in an estimate of the Fischer information matrix.

Following \cite{wooldridge2010econometric}, we estimate $A_0$ and $B_0$ as \begin{align}\hat{A}_0 &= -\frac{1}{N}\sum_{t=1}^T \sum_{k=1}^N \sum_{i=1}^n \nabla_\theta \ell_{it}^k (\theta) \left( \nabla_\theta \ell_{it}^k (\theta) \right)^\top\\ \hat{B_0} &= \hat{A}_0 + \frac{1}{N} \sum_{k=1}^N \sum_{t=1}^T \sum_{r\neq t} \sum_{i=1}^n \nabla_\theta \ell_{it}^k (\theta) \left( \nabla_\theta \ell_{ir}^k (\theta) \right)^{\top}\end{align} where the additional cross products in $\hat{B}_0$ compensate for potential serial correlation. Finally, the asymptotic variance matrix is estimated as $\frac{1}{N}\hat{A}_0\inv \hat{B}_0 \hat{A}_0\inv$.

Because of the reasonably limited size of our dataset, we also construct p-values through a fully nonparametric bootstrap based on resampling of full group sequences. The nonparametric clustered bootstrap approach works as its name suggests. Sequences are drawn from the sample (with replacement) to form a new sample, matching the total number of observations and proportion of observations of treatment groups. In order to ensure proper clustering, it is crucial that each sequence of networks from one group be treated as a single observation. The estimation routine is then run again on this new dataset, and the estimates are stored. Standard errors are obtained as the standard deviation of estimates made by repeating this process many times. This approach leverages the argument that drawing from the observed sequences with replacement is the best approximation we can make to drawing from the true distribution of sequences without making any distributional assumptions. As such, the nonparametric bootstrap should generally provide a more conservative estimate of the true variability of the estimators.

Further, we do not expect the MLE to reach its asymptotic distribution within a limited sample size. Therefore, we make no assumptions on the normality of the estimator evaluated on our sample, which would be implicit in using the standard error as our metric of uncertainty and suggesting that it would imply t-statistics or p-values. Instead, we present nonparametric estimates of the p-values which are computed directly from the bootstrap results as the empirical probability that the estimator takes the opposite sign to what was estimated on the full sample.

\subsection{Estimation Results} We turn next to the structural estimation of the discrete choice framework presented above. We first estimate the following simple specification of behavior, in order to understand how the treatment changes incentives by directing attention toward direct and generalized reciprocity, (as in \cite{fehr1997reciprocity,fehr1998reciprocity}). We also incorporate some persistence parameters to account for dynamic effects or inertia in decisions:\begin{align*}
    \beta_{it} (A,c | \theta) &= \underbrace{\theta_2T \kappa_i(A_{-i},c_{-i})}_\text{Change in altruism} + \underbrace{\frac{\theta_5 T}{n-1} \sum_{j=1,j\neq i}^n A_{ij,t}A_{ji,t-1} m_i(A_{i,t})m_j(A_{j,t-1}) c_{i,t}c_{j,t-1}}_\text{Reciprocity} \\&+ \underbrace{\frac{\theta_3 + \theta_4 T}{(n-1)^2} \sum_{j=1,j\neq i}^n A_{ij,t} m_i(A_{i,t})c_{i,t}\sum_{k=1,k\neq i}^n A_{ki,t-1}m_i(A_{k,t-1})c_{k,t-1}}_{\text{Generalized reciprocity}}
    \\&+ \underbrace{\theta_6 (c_{i,t}-c_{i,t-1})^2}_{\text{Contribution inertia}} +\underbrace{\theta_7\sum_{j=1,j\neq i}^n |A_{ij,t} - A_{ij,t-1}|}_{\text{Link inertia}}
\end{align*}

\begin{table}[!t]
    \centering
    \small
    \caption{\textbf{Structural parameter estimates}}
    \label{res:struct1}
    \begin{tabular}{l*{3}{c}}
        \hline\\
        \multicolumn{1}{c}{} & \multicolumn{3}{c}{MLE}\\
        & \multicolumn{1}{c}{Estimate} & \multicolumn{1}{c}{SE} & \multicolumn{1}{c}{Bootstrap p-val}\\
        \hline\\
        ($\theta_1$) Contribution costs & 9.355 & 0.107 & 0.000\\[1em]
        ($\theta_2$) Treatment $\times$ & 2.629 & 0.224 & 0.214\\
        \hspace{1.8em} Contribution costs \\
        ($\theta_3$) Generalized reciprocity & 7.011 & 0.090 & 0.005\\[1em]
        ($\theta_4$) Treatment $\times$ & -9.177 & 0.279 & 0.038\\
        \hspace{1.8em} Generalized reciprocity\\
        ($\theta_5$) Treatment $\times$ & 26.155 & 0.382 & 0.000\\
        \hspace{1.8em} Reciprocity \\
        ($\theta_6$) Contribution inertia & -5.082 & 0.038 & 0.000\\[1em]
        ($\theta_7$) Link inertia & -0.542 & 0.009 & 0.000\\[1em]       
        \hline
        Number of groups & 46 & 46 & 46\\
        Number of treatment groups & 28 & 28 & 28\\
        Number of time periods & 30 & 30 & 30\\
        \hline
        Bootstrap samples & & & 1000\\
        \hline
        \multicolumn{4}{l}{Asymptotic standard errors clustered by group}
    \end{tabular}
\end{table}

We estimate this model using the data from our laboratory experiment and present the results of these estimations in Table \ref{res:struct1}.

We focus on the estimates of five primary parameters.\footnote{Both of the inertia terms are significantly negative, indicating a tendency of players to bias their actions toward those that they took in the previous round.} The first parameter, $\theta_1$, captures how players weigh the cost of sharing. Not surprisingly, all estimates of this parameter are positive. This means that players follow their incentives and that, holding all else constant, they are more likely to choose actions that are less costly.

The other four main coefficients concern the behavioral component of payoffs. The positive sign of $\theta_2$, which is the coefficient on the interaction between contribution cost and the treatment, indicates that having access to this new information makes players more careful about who they share with. We can interpret this coefficient as a reduction in altruism due to the treatment since it indicates a tendency of players to focus more on their individual costs of sharing. 

The coefficients $\theta_3$ and $\theta_4$ are related to generalized reciprocity. The interpretation of generalized reciprocity is that it measures the tendency of players to share more when they receive more overall benefits as a result of their group's sharing in the previous round, without regard to who exactly shared with them. The estimate of $\theta_3$ is positive indicating that in the baseline sessions, generalized reciprocity drives some of the observed sharing behavior. The negative coefficient on the treatment effect ($\theta_4$), however, drives behavior in the opposite direction. Because the sum of effects is significantly negative, this suggests a strict tradeoff between generalized and direct reciprocity, and highlights how the information provided in the treatment helps to focus reciprocity toward active collaborators.

The final coefficient of primary interest ($\theta_5$) is on the interaction between the treatment effect and our measure of direct reciprocity. The large magnitude and significance of this estimate indicate that direct reciprocity is indeed a strong driver of the more successful collaboration observed in treatment groups. 

\subsection{Individual heterogeneity\label{ss:heterogeneity}}

\begin{figure}[!t]
    \centering
    \includegraphics[width=0.75\linewidth]{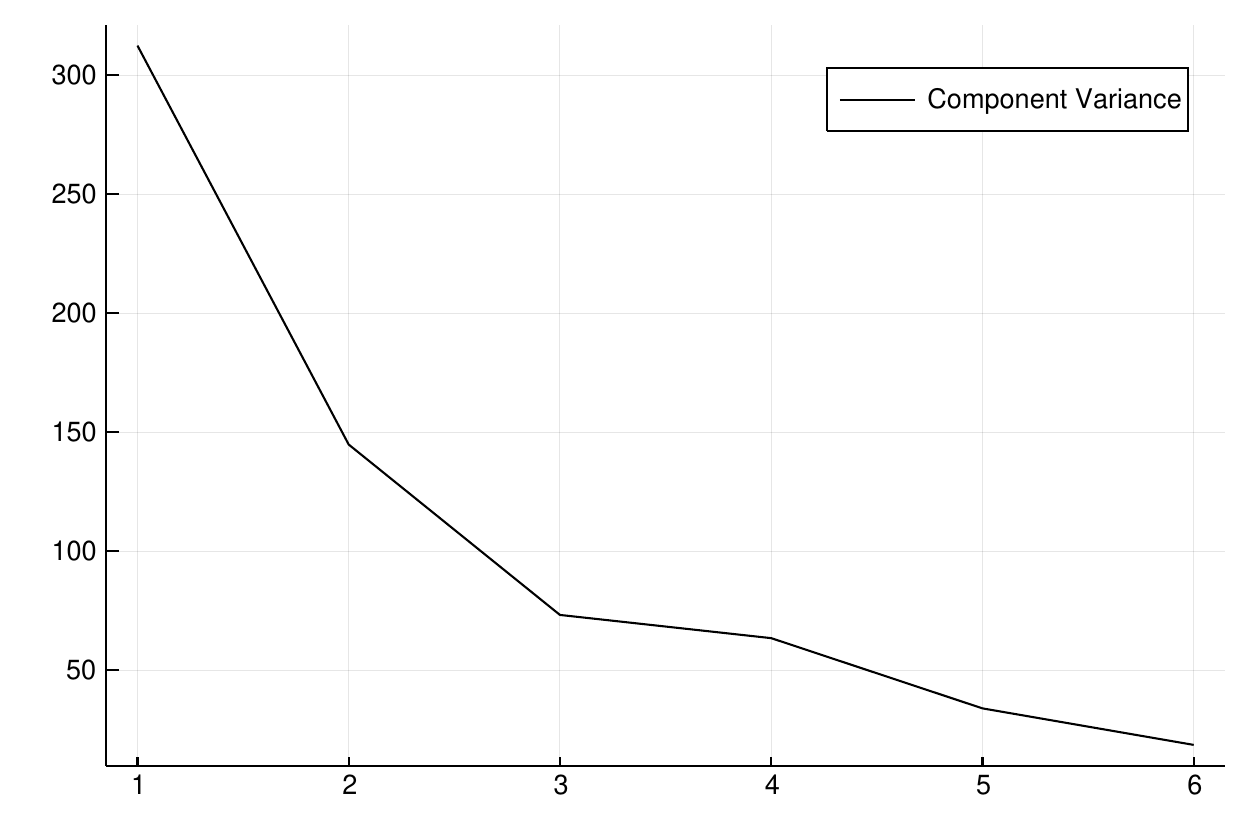}
    \caption{\textbf{Scree plot of the eigenvalues determining reciprocity characteristics}}
    \label{fig:scree}
\end{figure}

After the main parts of the experiment, subjects were asked a series of survey questions to elicit measures of their individual characteristics. All of these questions were aimed at eliciting a subject's heterogeneous preferences toward trust and reciprocity. The full set of survey questions is reproduced in Appendix \ref{app:questions}. In order to distill answers from these questions into a reasonable number of attributes to describe reciprocity and trust, we performed a principal components analysis and selected the top two principal components as representative characteristics of reciprocity, along with the first principal component of the trust questionnaire responses to represent trust. The decision to select two components for reciprocity was determined by referring to a Scree plot of the singular values shown in Figure \ref{fig:scree}. From this, we can see that the first two components of the reciprocity questionnaire describe more than $70\%$ of the variation in survey responses. The first principal component, referred to as overall reciprocity, captures both positive and negative reciprocity, while the second principal component, referred to as positive reciprocity, favors positive reciprocity and eschews negative reciprocity.

We estimate the same model as above with additional interaction terms between the main parameters of interest and the subject's individual characteristics (the trust attribute, overall reciprocity attribute, and positive reciprocity attribute). The results are reported in Table~\ref{res:struct2}. We find strong evidence that individually heterogeneous preferences drive differences in behavior. Importantly, the point estimates for $\theta_{1}$--$\theta_{7}$ are robust to the introduction of these heterogeneous individual characteristics. Tables~\ref{tab:recpc} and \ref{tab:trustpc} show the weights assigned to each response by the relevant components, with negative weights highlighted in grey.

\begin{table}[ht]\caption{\textbf{Structure of reciprocity characteristics from PCA}\label{tab:recpc}}\scriptsize
\begin{tabular}{|l|c|c|}
\hline
\multicolumn{1}{|c|}{}                                                                       & \multicolumn{2}{c|}{PCA Component Weights}     \\ \cline{2-3} 
Prompt                                                                                         & \multicolumn{1}{c|}{Component 1} & Component 2 \\ \hline
If someone does me a favor, I am prepared to &  & 
\\return it.                                        & \multirow{-2}{*}{0.3737}     & \multirow{-2}{*}{0.4665}      \\ \hline
If I suffer a serious wrong, I will take revenge &  & \cellcolor[gray]{.9}
\\ as soon as possible, no matter what the cost. & \multirow{-2}{*}{0.4502}      & \multirow{-2}{*}{\cellcolor[gray]{.9}-0.3793}     \\ \hline
If somebody puts me in a difficult position, I &  & \cellcolor[gray]{.9}
\\ will do the same to them.                       & \multirow{-2}{*}{0.4622}      & \multirow{-2}{*}{\cellcolor[gray]{.9}-0.4069}   \\ \hline
I go out of my way to help somebody who has & & 
\\ been kind to me before.                            & \multirow{-2}{*}{0.3656}      &\multirow{-2}{*}{0.4515}      \\ \hline
If somebody insults or offends me, I will & & \cellcolor[gray]{.9}
 \\ offend or insult them back.                          & \multirow{-2}{*}{0.4345}      &\multirow{-2}{*}{\cellcolor[gray]{.9} -0.2966}     \\ \hline
I am ready to undergo personal costs to help &  & 
\\ somebody who has helped me before.                & \multirow{-2}{*}{0.3487}      & \multirow{-2}{*}{0.4255}     \\ \hline
\end{tabular}
\end{table}
\begin{table}[ht]\caption{\textbf{Structure of trust characteristics from PCA}\label{tab:trustpc}}\scriptsize
\begin{tabular}{|l|c|}
\hline
Prompt                                                                         & Component 1 \\ \hline
In general, one can trust people.                                              & 0.5787                \\ \hline
These days you cannot rely on anybody else.                                    & \cellcolor[gray]{.9}-0.6485               \\ \hline
When dealing with strangers, it is better to be careful before you trust them. & \cellcolor[gray]{.9}-0.4945               \\ \hline
\end{tabular}
\end{table}

\begin{table}[!t]
    \centering
    \small
    \caption{\textbf{Structural parameter estimates with heterogeneity}}
    \singlespacing
    \label{res:struct2}
    \scriptsize
    \begin{tabular}{l c c c}
        \hline
        & \multicolumn{3}{c}{MLE} \\
        & Estimate & SE & \multicolumn{1}{c}{Bootstrap p-val}\\
        \hline\\
        Contribution costs & 10.471 & 0.110 & 0.000\\

        T $\times$ Contribution costs & 2.480 & 0.229 & 0.241\\
        Generalized reciprocity & 9.385 & 0.084 & 0.003\\
        T $\times$ Generalized reciprocity & -13.033 & 0.305 & 0.002\\
        T $\times$ Reciprocity & 25.039 & 0.415 & 0.000\\
        Contribution inertia & -5.626 & 0.037 & 0.000\\
        Link inertia & -0.516 & 0.008 & 0.000\\[1em]
        Individual Heterogeneity \\
        \hline\\
        Trust $\times$ Contribution costs & 1.951 & 0.055 & 0.011\\
        Overall reciprocity $\times$ Contribution costs & -2.966 & 0.066 & 0.051\\
        Positive reciprocity $\times$ Contribution costs & -6.916 & 0.117 & 0.002\\[1em]
        Trust $\times$ T $\times$ Contribution costs & -1.951 & 0.123 & 0.091\\
        Overall reciprocity $\times$ T $\times$ Contribution costs & 4.413 & 0.143 & 0.052\\
        Positive reciprocity $\times$ T $\times$ Contribution costs & 4.266 & 0.131 & 0.068\\[1em]
        Trust $\times$ Generalized reciprocity & 2.128 & 0.048 & 0.099\\
        Overall reciprocity $\times$ Generalized reciprocity & 4.286 & 0.064 & 0.065\\
        Positive reciprocity $\times$ Generalized reciprocity & -8.645 & 0.130 & 0.021\\[1em]
        Trust $\times$ T $\times$ Generalized reciprocity & 0.347 & 0.277 & 0.428\\
        Overall reciprocity $\times$ T $\times$ Generalized reciprocity & -6.300 & 0.330 & 0.052\\
        Positive reciprocity $\times$ T $\times$ Generalized reciprocity & 9.356 & 0.314 & 0.121\\[1em]
        Trust $\times$ T $\times$ Reciprocity & 3.151 & 0.392 & 0.167\\
        Overall reciprocity $\times$ T $\times$ Reciprocity & 3.289 & 0.373 & 0.149\\
        Positive reciprocity $\times$ T $\times$ Reciprocity & -2.345 & 0.382 & 0.302\\[1em]
        \hline
        Number of groups & 46 \\
        Number of treatment groups & 28\\
        Number of time periods & 30 \\
        Number of bootstrap samples & & & 1000 \\
        \hline
        \multicolumn{4}{l}{Asymptotic standard errors and bootstrap p-values clustered for 46 groups.} \\
    \end{tabular}
\end{table}

Because there are only three trust questions, the first principal component summarizes most of the information from the trust questionnaire. It places positive weight on the question that involves trust and negative weights on two questions that suggest mistrust. Perhaps surprisingly, this measure of trust is associated with a positive interaction on contribution costs in the baseline, which indicates that individuals who score highly on trust are less altruistic and more careful about where they direct effort in the baseline. This agrees with the results of \cite{glaeser2000measuring}, which suggest that such trust questionnaires predict trustworthy behavior but do not necessarily predict trusting behavior. Further in line with these results is a strong positive interaction of the trust characteristic with generalized reciprocity in the baseline. This suggests that these individuals are trustworthy in that they respond to sharing by others by increasing their own contribution. However, they are less likely to share blindly and trust that others will reciprocate. In the treatment, estimates of the effect of trust are less precise but suggest a reversal of this phenomenon; they trust that others will reciprocate when they know that others will be aware of their sharing behavior. This is captured by the negative estimate of the interaction between trust, the treatment indicator, and contribution costs, together with the positive estimate of the coefficient for the interaction between trust, the treatment indicator, and direct reciprocity. This sheds more light on information as a mechanism driving the mixed results regarding trust and sharing behavior in public goods games, observed in previous work \citep{anderson2004social}.

The characteristic that we describe as overall reciprocity consists of positive weights on the answers to all of the questions in the reciprocity questionnaire. This includes both questions about positive reciprocity (e.g. ``If someone does me a favor, I am prepared to return it''), as well as negative reciprocity (``If someone puts me in a difficult position, I will do the same to them''). Estimates of the interaction between this characteristic and the behavioral utility terms suggest that these individuals are more altruistic in the baseline and behave more in line with generalized reciprocity. At the onset of the treatment, they also shift more weight toward direct reciprocity. However, this shift toward direct reciprocity is potentially offset by a decrease in altruism (measured by additional weight placed on the costs of contributing) coupled with a strong decrease in generalized reciprocity. This suggests that individuals who have a high overall reciprocity attribute use new information to discriminate between collaborators as a mechanism for punishment.

On the other hand, the second component of reciprocity places positive weight on questions involving positive reciprocity and negative weight on questions involving negative reciprocity or punishment. Individuals who align with this characteristic place much lower weight on the actual cost of contributing, suggesting some altruism. While there is some tradeoff in the treatment, the sign of the aggregate interaction term remains negative in the treatment suggesting that these players are still behaving more altruistically than average. Perhaps surprisingly, there is a strong negative coefficient on the interaction between positive reciprocity and generalized reciprocity in the baseline. These together suggest that their increased sharing is not conditional on having received more benefits from their group, possibly representing a tendency to share in anticipation that others will behave reciprocally. This interpretation is reinforced by a large positive effect of the treatment on generalized reciprocity for this group, offset by a small decrease in direct reciprocity. In other words, these individuals reciprocate by sharing with the entire group, and trusting in the reciprocity of others, rather than by using new information as a tool for punishment.

\subsection{Simulations and goodness of fit}

To address concerns about how well the structural model fits the data, and to enable our counterfactual investigations involving behavioral interventions, we start by conducting some simulations. The simulations are conducted semi-parametrically, by bootstrap sampling initial conditions with replacement and using them as the starting point for a sequence of network and contribution profiles, drawn with Markov-Chain Monte Carlo (MCMC) sampling. Although a full MCMC chain would be slow to mix, we can exploit conditional independence relationships, due to the panel structure, to build a partial Gibbs sampler that draws from each player's next strategies individually---treating a tuple $(A_i,c_i)$ individually using a combination of MCMC and rejection sampling. Whenever possible, we avoid distributional assumptions and opt for semi-parametric procedures.

The partial Gibbs MCMC sampler works as follows; first we draw a starting panel from an initial distribution $f_0(A,c)$. For this, we resample from the starting panels observed in the data either in the treatment or control. Next, we iteratively draw each player's next strategy using a standard Metropolis-Hastings sampler. This leverages the fact that players' strategies are independent, conditional on the previously observed network. Using a proposal distribution $(A^m_i c_i^m) \sim g_i (A_i, c_i)$, the acceptance ratio for sample $m$ is determined by $\alpha = \frac{\exp\left( \phi_i (A^m_i, c^m_i \mid \Omega_{it}, \theta )\right)}{\exp\left( \phi_i (A^{m-1}_i, c^{m-1}_i \mid \Omega_{it}, \theta ) \right)} \frac{g_i (A_i^{m-1},c_i^{m-1})}{g_i (A_i^m, c_i^m)}$, and sample $m$ is accepted with probability equal to $\min \{ 1, \alpha\}$. Due to the standard theory of Metropolis-Hastings sampling, we then have $\lim_{m\rightarrow \infty} (A^m_i, c_i^m) \sim f_i(A_i^m, c_i^m \mid \Omega_{it}^k, \theta)$. In practice, we choose a large number $M$ and choose the action $(A^M_i, c^M_i)$ to be included in the next network panel.\footnote{For this paper, we used a sample length of 10,000 in between accepted actions.} This is repeated for each of the $n$ players, and the resulting panel $(A^M,c^M)$ is then taken to be the next step in the sequence, and this process is then repeated for each of the $T$ time periods to generate a single full sequence. This partial Gibbs sampler (that is, drawing each player's strategy variables separately) greatly improves the number of accepted samples, and thus the speed of mixture, by dramatically decreasing the magnitude of the sample space and thus increasing acceptance probabilities. 

In practice, because we want the proposal distribution to be close to the true distribution, we draw from a proposal distribution $g_i$ that is constructed as follows; the empirical distribution of contributions in the baseline and treatment sessions along with a random binary vector of links with equal probability on each feasible configuration. That is,\begin{equation}g_i (A^m_i, c^m_i) = \begin{cases}
    \frac{1}{2^{n-1}-1} \frac{1}{nNT}\sum_{t=1}^T \sum_{k=1}^N \sum_{i=1}^n \mathbb{I}[c_{it}^k = c_i^m] & \text{ if } c_i^m \neq 0 \text{ and } A_i^m \neq \mathbf{0}\\
    \frac{1}{nNT}\sum_{t=1}^T \sum_{k=1}^N \sum_{i=1}^n \mathbb{I}[c_{it}^k = 0] & \text{ if } c_i^m = 0 \text{ and } A_i^m = \mathbf{0}\\
    0 &\text{ otherwise}
\end{cases}\end{equation}

After drawing a sample of 1000 simulated groups, 500 in the treatment and 500 in the control, we evaluate the fit of the model by replicating our reduced form estimations from Table~\ref{res:reduced}. The results of this replication for the MLE estimates are shown in Tables~\ref{res:mle_gof} and \ref{res:mle_gof_h}. These evaluations of the goodness-of-fit of the structural model are entirely promising---point estimates match up nearly perfectly in sign and closely in magnitude to the reduced form static and dynamic treatment effects. A further look at the simulations in Figures~\ref{fig:counterfactualstrust}, \ref{fig:counterfactualsrec1}, and \ref{fig:counterfactualsrec2} reinforces this finding. In these panels, the cross-group averages of 10,000 simulated groups for each condition are shown as pale lines overlaying the dashed lines (which show the patterns from the original data.) These figures highlight the ability of this structural model to reproduce observed patterns of behavior on all of our key metrics. Crucially, because the vast majority of these metrics (with the sole exception of individual costs) are not directly fit by the structural model, this provides some confidence that the model is accurately capturing key features of subject behavior.

\begin{table}[!t]\small\caption{\textbf{Reduced form replication under MLE estimates}\label{res:mle_gof}}
   {
\def\sym#1{\ifmmode^{#1}\else\(^{#1}\)\fi}\scriptsize
\begin{tabular}{l*{6}{c}}
\hline\\
            &\multicolumn{1}{c}{(1)}&\multicolumn{1}{c}{(2)}&\multicolumn{1}{c}{(3)}&\multicolumn{1}{c}{(4)}&\multicolumn{1}{c}{(5)}&\multicolumn{1}{c}{(6)}\\
            &\multicolumn{1}{c}{Efficient structure}&\multicolumn{1}{c}{Contributions}&\multicolumn{1}{c}{Links}&\multicolumn{1}{c}{Costs}&\multicolumn{1}{c}{Reciprocity}&\multicolumn{1}{c}{Centralization}\\
\hline\\
Period      &    -0.00434&    -0.00896&    -0.00297&    -0.00193&   -0.000319&     0.00102\\
            &  (0.000526)&  (0.000746)&   (0.00211)&  (0.000196)& (0.0000287)&  (0.000191)\\
[1em]
Treatment   &       0.459&       1.113&       1.969&       0.342&      0.0391&      -0.110\\
            &    (0.0363)&    (0.0554)&     (0.126)&    (0.0155)&   (0.00567)&   (0.00923)\\
[1em]
T $\times$ Period&    -0.00803&     -0.0113&     -0.0467&    -0.00470&    0.000195&     0.00175\\
            &   (0.00164)&   (0.00266)&   (0.00583)&  (0.000741)&  (0.000253)&  (0.000438)\\
[1em]
Group FE      &       Yes&      Yes&       Yes&     Yes&      Yes&      Yes\\       \\
\hline
\(N\)       &        30000         &        30000         &        30000         &        30000         &        30000         &        30000         \\
\hline
\multicolumn{7}{l}{\footnotesize Standard errors in parentheses, clustered for 1000 groups}\\
\end{tabular}
}
\end{table}

\begin{table}[!t]\small\caption{\textbf{Reduced form replication from MLE with heterogeneity}\label{res:mle_gof_h}}
   {
\def\sym#1{\ifmmode^{#1}\else\(^{#1}\)\fi}\scriptsize
\begin{tabular}{l*{6}{c}}
\hline\\
            &\multicolumn{1}{c}{(1)}&\multicolumn{1}{c}{(2)}&\multicolumn{1}{c}{(3)}&\multicolumn{1}{c}{(4)}&\multicolumn{1}{c}{(5)}&\multicolumn{1}{c}{(6)}\\
            &\multicolumn{1}{c}{Efficient structure}&\multicolumn{1}{c}{Contributions}&\multicolumn{1}{c}{Links}&\multicolumn{1}{c}{Costs}&\multicolumn{1}{c}{Reciprocity}&\multicolumn{1}{c}{Centralization}\\
\hline\\
Period           &    -0.00217&    -0.00422&     0.00364&   -0.000567&   -0.000156&    0.000735\\
            &  (0.000581)&  (0.000893)&   (0.00213)&  (0.000250)& (0.0000389)&  (0.000181)\\
[1em]
Treatment         &       0.449&       0.876&       1.702&       0.274&      0.0379&     -0.0911\\
            &    (0.0340)&    (0.0536)&     (0.123)&    (0.0156)&   (0.00576)&   (0.00836)\\
[1em]
T $\times$ Period       &    -0.00810&   -0.000625&     -0.0412&    -0.00152&    0.000845&     0.00104\\
            &   (0.00163)&   (0.00271)&   (0.00587)&  (0.000779)&  (0.000246)&  (0.000406)\\
[1em]
Group FE      &       Yes&      Yes&       Yes&     Yes&      Yes&      Yes\\       \\
\hline
\(N\)       &        30000         &        30000         &        30000         &        30000         &        30000         &        30000         \\
\hline
\multicolumn{7}{l}{\footnotesize Standard errors in parentheses, clustered for 1000 groups}\\
\end{tabular}
}
\end{table}

\subsection{Counterfactual behavioral interventions}\label{ss:counterfactuals}

In our estimated model, it is clear that behavioral heterogeneity features (trust, overall reciprocity, positive reciprocity) drive changes in collaboration patterns across the experiments. In a set of counterfactual simulations, we examine the effects of uniform upward shifts in each one of these characteristics across the population. These can be described as ``benign behavioral interventions'', because they mimic the effects of priming players toward a certain social preference but without directly altering payoffs or the information structure. We use these counterfactuals to reason about which behavioral traits would be the most valuable to promote collaboration.

From a platform design perspective, we can think of these as policy counterfactuals to investigate marketing strategies that would emphasize certain behavioral traits. More generally, due to the complex interactions between these characteristics and the different forms of altruism and reciprocity, it is useful to know which behavioral traits are the most valuable in promoting collaboration. The results, shown in Figures~\ref{fig:counterfactualstrust}, \ref{fig:counterfactualsrec1}, and \ref{fig:counterfactualsrec2}, highlight the effects of these counterfactual adjustments. In each figure, we show the average of 10,000 simulations in the counterfactual regime, shown as the darker lines, juxtaposed against 10,000 simulations from the standard model as well as the values actually observed in the data (shown as a dashed line).

\begin{figure}[!t]
    \centering
    \includegraphics[width=\linewidth]{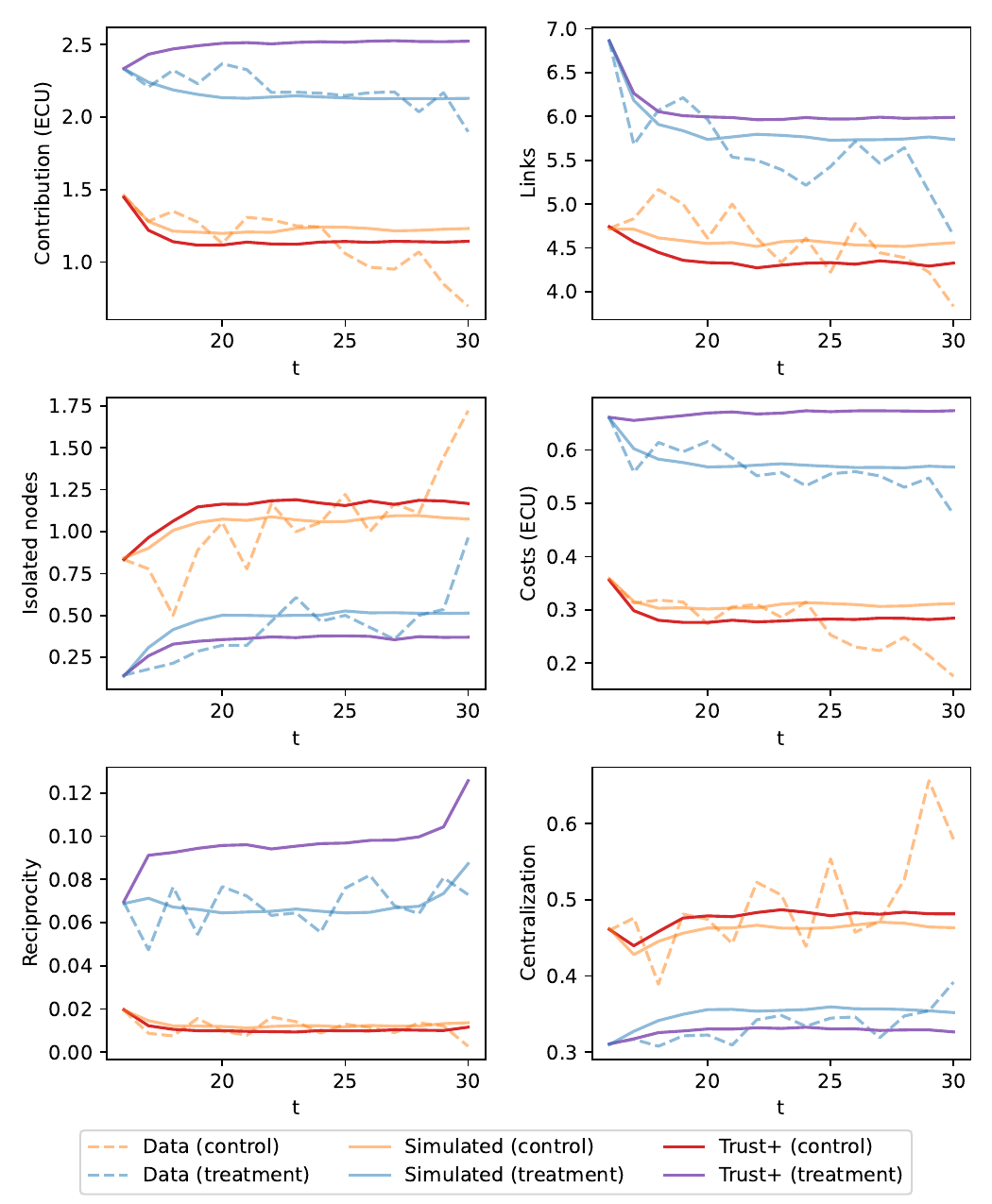}
    \caption{\textbf{Results of counterfactual uniform increases to trust}}
    \label{fig:counterfactualstrust}
\end{figure}\clearpage

\begin{figure}[!t]
    \centering
    \includegraphics[width=\linewidth]{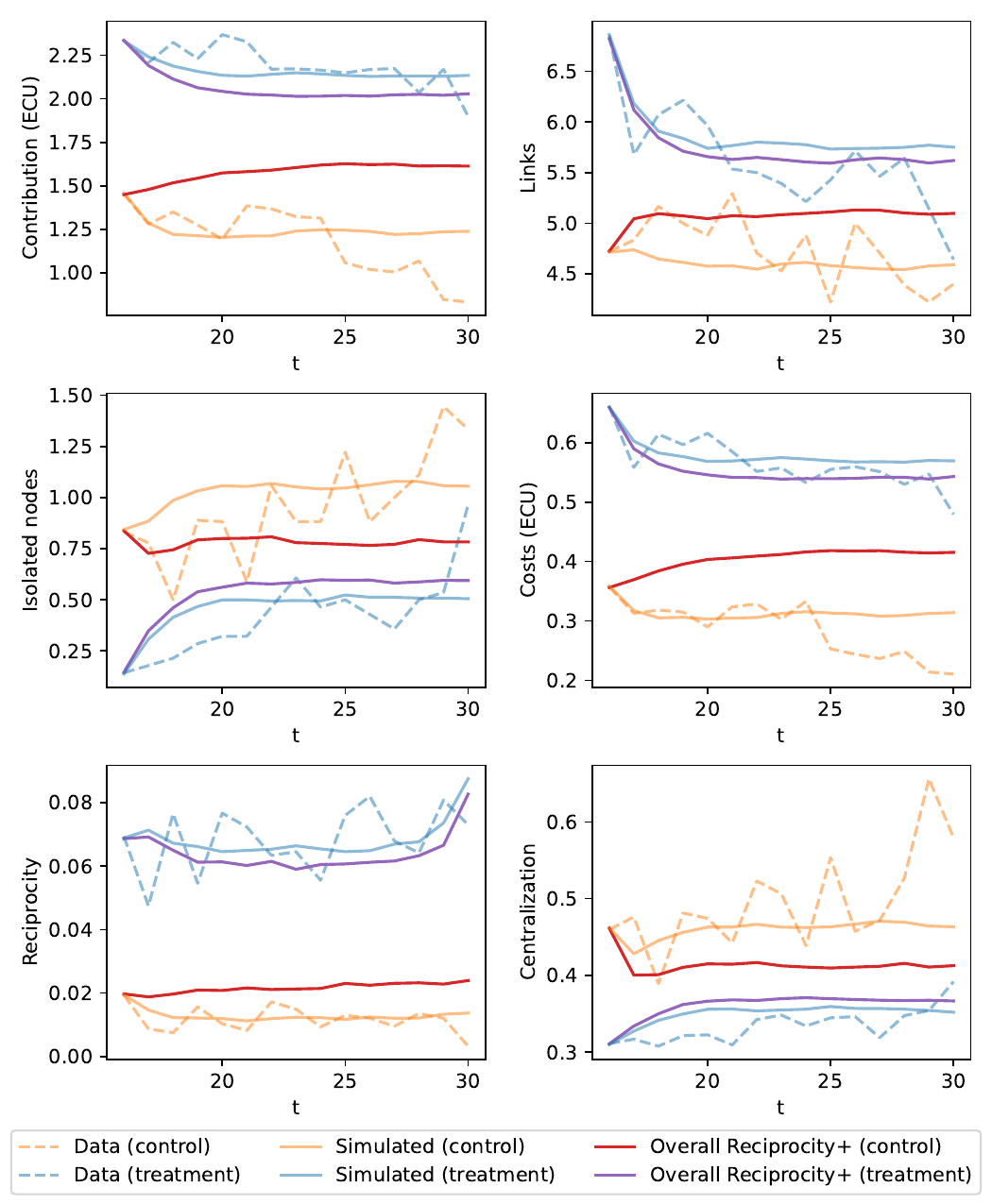}
    \caption{\textbf{Results of counterfactual uniform increases to overall reciprocity}}
    \label{fig:counterfactualsrec1}
\end{figure}\clearpage

\begin{figure}[!t]
    \centering
    \includegraphics[width=\linewidth]{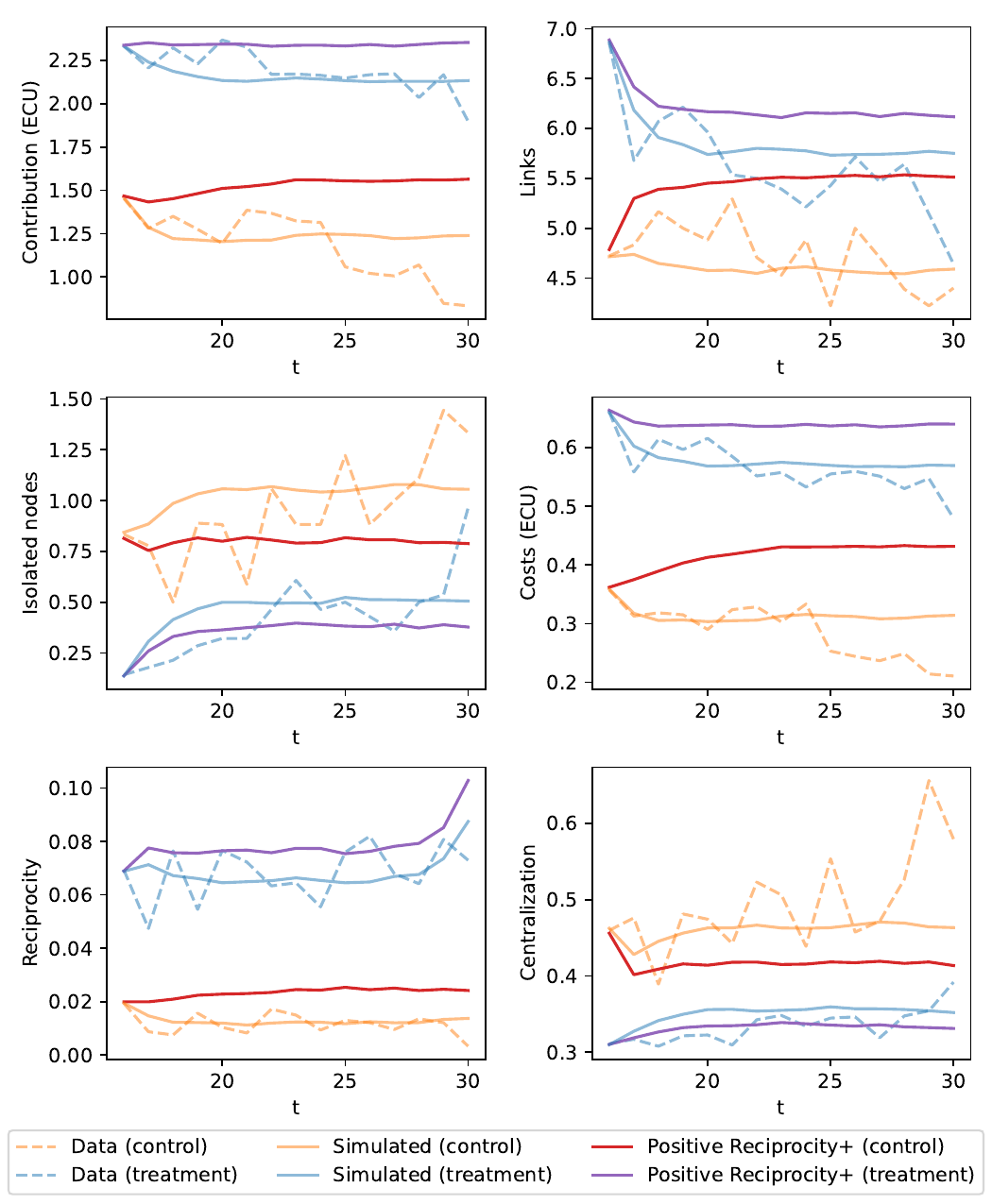}
    \caption{\textbf{Results of counterfactual uniform shifts toward positive reciprocity and away from negative reciprocity}}
    \label{fig:counterfactualsrec2}
\end{figure}\clearpage

Figure~\ref{fig:counterfactualstrust} shows the effect of a one-standard deviation increase in the trust component, applied uniformly across the population. As the figure shows, this intervention improves outcomes along all measures, including contributions, links, reciprocity, and centralization, in the treatment condition. In the baseline, however, the effect is moderately reversed. This can be explained by the reduction in altruism and in the overall willingness to internalize sharing costs, as well as by an increase in free-riding to take advantage of those who are more trusting.

The next set of simulations, shown in Figure~\ref{fig:counterfactualsrec1}, shows the effects of a uniform increase in overall reciprocity. From this figure, we can see that our measure of overall reciprocity is essentially capturing the opposite effect as trust; in the baseline, overall reciprocity creates substantial increases in linking and contributions. In the treatment, however, when subjects are given access to information that can be used as a tool for punishment, observed networks are actually less efficient than before the change.

Finally, Figure~\ref{fig:counterfactualsrec2} shows the effects of an increase in positive reciprocity, and thus a shift away from a punishment mindset and toward a focus on gains from mutual effort and collaboration. As we can see, shifting players toward positive reciprocity has some of the largest effects in the baseline, with large increases in linking behavior and increases in contributions that are similar to those from a shift in overall reciprocity. Unlike the other examples, however, the increases from positive reciprocity are consistent across both the baseline and the treatment, generating substantial efficiency gains in both settings. From this perspective, it seems that positive reciprocity is the most consistently valuable of these social preferences if the goal is to encourage efficient collaboration.

\section{Conclusion\label{sec:conclusion}}

In this paper, we develop a model of collaboration in the form of resource sharing on an endogenous network. Unlike previous models of voluntary resource sharing, we allow individuals to specifically choose the beneficiaries of their externalities. This intuitive extension provides a convenient generalization of the classical voluntary contributions mechanism and allows for flexible specification of group size effects in a way designed to capture realistic incentive problems.

We characterize simple theoretical and empirical frameworks for the voluntary sharing game and use them to estimate social preferences from panels of dynamic network data. Using data collected in the lab, we examine the impact of different information structures on sharing and link formation decisions, with a particular focus on the relative importance of different forms of trust and reciprocity. We find that individuals in these environments are highly effective at using information to coordinate on more efficient outcomes. While this form of collaboration is, theoretically, plagued by free-riding incentives, behavioral biases such as a preference for reciprocity can generate complementarities that help support nontrivial collaborative outcomes.

When provided with more detailed information, players exhibit a clear preference to target the positive externalities they generate toward others who have shared with them, reciprocating directly. Using our structural estimation methods, we characterize the tradeoff between altruism and different forms of reciprocity. We find that subjects tend to rely heavily on direct reciprocity when it is available, substituting away from a more generalized version of reciprocity. We also find that the subjects' reliance on different types of reciprocity is correlated strongly with heterogeneous behavioral characteristics created by a feature analysis of survey questions collected during the experiment. Using the interactions between these features and the different forms of reciprocity, we simulate counterfactual increases in different types of social preferences on the platform. The effects of these interventions vary strongly based on the information environment. Increasing trust can backfire in the baseline by increasing the incentive to free-ride, and increasing overall reciprocity has the opposite effect of driving efficiency gains in the baseline but falling short in the treatment due to increases in punishment behavior. On the other hand, interventions that focus on pushing players toward positive reciprocity (returns from mutual cooperation) and away from negative reciprocity (punishment) have a consistent effect of boosting collaboration and efficiency.

While we have demonstrated the effectiveness of our modeling paradigm and established proof-of-concept, the real power of our methodology is in its capacity to explain far more sophisticated patterns of learning and behavior. In particular, using laboratory experiments to collect small panels of network data opens the door to a new family of statistical tools that can be used to identify the structure of social preferences, and which have been understudied due to the massive size of most network data. In this respect, we hope that our methodology and our reported findings will lay the groundwork for future directions of study, including but not limited to a better understanding of the motivations for and patterns of collaboration that guide the formation and evolution of prosocial behavior in the sharing economy.

\bibliographystyle{te}
\singlespacing
\bibliography{bibliography}

\begin{thebibliography}{40}
\newcommand{\enquote}[1]{``#1''}
\providecommand{\natexlab}[1]{#1}
\providecommand{\url}[1]{\texttt{#1}}
\providecommand{\urlprefix}{URL }
\providecommand{\bibAnnoteFile}[1]{%
  \IfFileExists{#1}{\begin{quotation}\noindent\textsc{Key:} #1\\
  \textsc{Annotation:}\ \input{#1}\end{quotation}}{}}
\providecommand{\bibAnnote}[2]{%
  \begin{quotation}\noindent\textsc{Key:} #1\\
  \textsc{Annotation:}\ #2\end{quotation}}

\bibitem[{Al{\'o}s-Ferrer and Netzer(2010)}]{alos2010logit}
Al{\'o}s-Ferrer, Carlos and Nick Netzer (2010), \enquote{The logit-response
  dynamics.} \emph{Games and Economic Behavior}, 68, 413--427.
\bibAnnoteFile{alos2010logit}

\bibitem[{Anderson et~al.(2004)Anderson, Mellor, and
  Milyo}]{anderson2004social}
Anderson, Lisa~R, Jennifer~M Mellor, and Jeffrey Milyo (2004), \enquote{Social
  capital and contributions in a public-goods experiment.} \emph{American
  Economic Review}, 94, 373--376.
\bibAnnoteFile{anderson2004social}

\bibitem[{Badev(2021)}]{badev2021nash}
Badev, Anton (2021), \enquote{Nash equilibria on (un) stable networks.}
  \emph{Econometrica}, 89, 1179--1206.
\bibAnnoteFile{badev2021nash}

\bibitem[{Boosey(2017)}]{boosey2017conditional}
Boosey, Luke~A (2017), \enquote{Conditional cooperation in network public goods
  experiments.} \emph{Journal of Behavioral and Experimental Economics}, 69,
  108--116.
\bibAnnoteFile{boosey2017conditional}

\bibitem[{Bramoull{\'e} et~al.(2007)Bramoull{\'e}, Kranton
  et~al.}]{bramoulle2007public}
Bramoull{\'e}, Yann, Rachel Kranton, et~al. (2007), \enquote{Public goods in
  networks.} \emph{Journal of Economic Theory}, 135, 478--494.
\bibAnnoteFile{bramoulle2007public}

\bibitem[{Brown(2024)}]{brown2024team}
Brown, Christopher~L (2024), \enquote{Team production in endogenous networks.}
  \emph{Journal of Economic Behavior \& Organization}, 217, 560--580.
\bibAnnoteFile{brown2024team}

\bibitem[{Camerer(2003)}]{camerer2003behavioural}
Camerer, Colin~F (2003), \enquote{Behavioural studies of strategic thinking in
  games.} \emph{Trends in cognitive sciences}, 7, 225--231.
\bibAnnoteFile{camerer2003behavioural}

\bibitem[{Chandrasekhar(2016)}]{chandrasekhar2016econometrics}
Chandrasekhar, Arun (2016), \enquote{Econometrics of network formation.}
  \emph{The Oxford Handbook of the Economics of Networks}, 303--357.
\bibAnnoteFile{chandrasekhar2016econometrics}

\bibitem[{Choi and Storr(2022)}]{choi2022market}
Choi, Ginny~Seung and Virgil~Henry Storr (2022), \enquote{The market as a
  process for the discovery of whom not to trust.} \emph{Journal of
  Institutional Economics}, 18, 467--482.
\bibAnnoteFile{choi2022market}

\bibitem[{Cox(2004)}]{cox2004identify}
Cox, James~C (2004), \enquote{How to identify trust and reciprocity.}
  \emph{Games and Economic Behavior}, 46, 260--281.
\bibAnnoteFile{cox2004identify}

\bibitem[{Dasaratha(2020)}]{dasaratha2020distributions}
Dasaratha, Krishna (2020), \enquote{Distributions of centrality on networks.}
  \emph{Games and Economic Behavior}, 122, 1--27.
\bibAnnoteFile{dasaratha2020distributions}

\bibitem[{Dasaratha(2023)}]{dasaratha2023innovation}
Dasaratha, Krishna (2023), \enquote{Innovation and strategic network
  formation.} \emph{The Review of Economic Studies}, 90, 229--260.
\bibAnnoteFile{dasaratha2023innovation}

\bibitem[{Elliott and Golub(2019)}]{elliott2019network}
Elliott, Matthew and Benjamin Golub (2019), \enquote{A network approach to
  public goods.} \emph{Journal of Political Economy}, 127, 730--776.
\bibAnnoteFile{elliott2019network}

\bibitem[{Fehr et~al.(2005)Fehr, Fischbacher, and
  Kosfeld}]{fehr2005neuroeconomic}
Fehr, Ernst, Urs Fischbacher, and Michael Kosfeld (2005),
  \enquote{Neuroeconomic foundations of trust and social preferences: initial
  evidence.} \emph{American Economic Review}, 95, 346--351.
\bibAnnoteFile{fehr2005neuroeconomic}

\bibitem[{Fehr and G{\"a}chter(1998)}]{fehr1998reciprocity}
Fehr, Ernst and Simon G{\"a}chter (1998), \enquote{Reciprocity and economics:
  The economic implications of homo reciprocans.} \emph{European Economic
  Review}, 42, 845--859.
\bibAnnoteFile{fehr1998reciprocity}

\bibitem[{Fehr and G{\"a}chter(2000)}]{fehr2000fairness}
Fehr, Ernst and Simon G{\"a}chter (2000), \enquote{Fairness and retaliation:
  The economics of reciprocity.} \emph{Journal of Economic Perspectives}, 14,
  159--182.
\bibAnnoteFile{fehr2000fairness}

\bibitem[{Fehr et~al.(1997)Fehr, G{\"a}chter, and
  Kirchsteiger}]{fehr1997reciprocity}
Fehr, Ernst, Simon G{\"a}chter, and Georg Kirchsteiger (1997),
  \enquote{Reciprocity as a contract enforcement device: Experimental
  evidence.} \emph{Econometrica: Journal of the Econometric Society}, 833--860.
\bibAnnoteFile{fehr1997reciprocity}

\bibitem[{Fischbacher(2007)}]{fischbacher2007z}
Fischbacher, Urs (2007), \enquote{z-tree: Zurich toolbox for ready-made
  economic experiments.} \emph{Experimental Economics}, 10, 171--178.
\bibAnnoteFile{fischbacher2007z}

\bibitem[{Fisher et~al.(1995)Fisher, Isaac, Schatzberg, and
  Walker}]{fisher1995heterogenous}
Fisher, Joseph, R~Mark Isaac, Jeffrey~W Schatzberg, and James~M Walker (1995),
  \enquote{Heterogenous demand for public goods: Behavior in the voluntary
  contributions mechanism.} \emph{Public Choice}, 85, 249--266.
\bibAnnoteFile{fisher1995heterogenous}

\bibitem[{Foster and Young(1990)}]{foster1990stochastic}
Foster, Dean and H~Peyton Young (1990), \enquote{Stochastic evolutionary game
  dynamics.} \emph{Theoretical Population Biology}, 38, 219--232.
\bibAnnoteFile{foster1990stochastic}

\bibitem[{Galeotti and Goyal(2010)}]{galeotti2010law}
Galeotti, Andrea and Sanjeev Goyal (2010), \enquote{The law of the few.}
  \emph{American Economic Review}, 100, 1468--92.
\bibAnnoteFile{galeotti2010law}

\bibitem[{Glaeser et~al.(2000)Glaeser, Laibson, Scheinkman, and
  Soutter}]{glaeser2000measuring}
Glaeser, Edward~L, David~I Laibson, Jose~A Scheinkman, and Christine~L Soutter
  (2000), \enquote{Measuring trust.} \emph{The Quarterly Journal of Economics},
  115, 811--846.
\bibAnnoteFile{glaeser2000measuring}

\bibitem[{Golub and Sadler(2021)}]{golub2021games}
Golub, Benjamin and Evan Sadler (2021), \enquote{Games on endogenous networks.}
  \emph{arXiv preprint arXiv:2102.01587}.
\bibAnnoteFile{golub2021games}

\bibitem[{Greiner(2015)}]{greiner2015subject}
Greiner, Ben (2015), \enquote{Subject pool recruitment procedures: Organizing
  experiments with orsee.} \emph{Journal of the Economic Science Association},
  1, 114--125.
\bibAnnoteFile{greiner2015subject}

\bibitem[{Gupta and Porter(2022)}]{gupta2022mixed}
Gupta, Harsh and Mason~A Porter (2022), \enquote{Mixed logit models and network
  formation.} \emph{Journal of Complex Networks}, 10, cnac045.
\bibAnnoteFile{gupta2022mixed}

\bibitem[{Hiller(2022)}]{hiller2022simple}
Hiller, Timo (2022), \enquote{A simple model of network formation with
  competition effects.} \emph{Journal of Mathematical Economics}, 99, 102611.
\bibAnnoteFile{hiller2022simple}

\bibitem[{Hommes and Ochea(2012)}]{hommes2012multiple}
Hommes, Cars~H and Marius~I Ochea (2012), \enquote{Multiple equilibria and
  limit cycles in evolutionary games with logit dynamics.} \emph{Games and
  Economic Behavior}, 74, 434--441.
\bibAnnoteFile{hommes2012multiple}

\bibitem[{Isaac and Walker(1988)}]{isaac1988group}
Isaac, R~Mark and James~M Walker (1988), \enquote{Group size effects in public
  goods provision: The voluntary contributions mechanism.} \emph{The Quarterly
  Journal of Economics}, 103, 179--199.
\bibAnnoteFile{isaac1988group}

\bibitem[{Jackson and Wolinsky(1996)}]{jackson1996strategic}
Jackson, Matthew~O and Asher Wolinsky (1996), \enquote{A strategic model of
  social and economic networks.} \emph{Journal of Economic Theory}, 71, 44--74.
\bibAnnoteFile{jackson1996strategic}

\bibitem[{Kinateder and Merlino(2017)}]{kinateder2017public}
Kinateder, Markus and Luca~Paolo Merlino (2017), \enquote{Public goods in
  endogenous networks.} \emph{American Economic Journal: Microeconomics}, 9,
  187--212.
\bibAnnoteFile{kinateder2017public}

\bibitem[{Kinateder and Merlino(2021)}]{kinateder2021evolution}
Kinateder, Markus and Luca~Paolo Merlino (2021), \enquote{The evolution of
  networks and local public good provision: A potential approach.}
  \emph{Games}, 12, 55.
\bibAnnoteFile{kinateder2021evolution}

\bibitem[{McKelvey and Palfrey(1995)}]{mckelvey1995quantal}
McKelvey, Richard~D and Thomas~R Palfrey (1995), \enquote{Quantal response
  equilibria for normal form games.} \emph{Games and Economic Behavior}, 10,
  6--38.
\bibAnnoteFile{mckelvey1995quantal}

\bibitem[{McKelvey and Palfrey(1998)}]{mckelvey1998quantal}
McKelvey, Richard~D and Thomas~R Palfrey (1998), \enquote{Quantal response
  equilibria for extensive form games.} \emph{Experimental Economics}, 1,
  9--41.
\bibAnnoteFile{mckelvey1998quantal}

\bibitem[{Mele(2017)}]{mele2017astructural}
Mele, Angelo (2017), \enquote{A structural model of dense network formation.}
  \emph{Econometrica}, 85, 825--850.
\bibAnnoteFile{mele2017astructural}

\bibitem[{Overgoor et~al.(2019)Overgoor, Benson, and
  Ugander}]{overgoor2019choosing}
Overgoor, Jan, Austin Benson, and Johan Ugander (2019), \enquote{Choosing to
  grow a graph: modeling network formation as discrete choice.} In \emph{The
  World Wide Web Conference}, 1409--1420, ACM.
\bibAnnoteFile{overgoor2019choosing}

\bibitem[{Page et~al.(1999)Page, Brin, Motwani, and
  Winograd}]{page1999pagerank}
Page, Lawrence, Sergey Brin, Rajeev Motwani, and Terry Winograd (1999),
  \enquote{The pagerank citation ranking: Bringing order to the web.} Technical
  report, Stanford InfoLab.
\bibAnnoteFile{page1999pagerank}

\bibitem[{Parise and Ozdaglar(2023)}]{parise2023graphon}
Parise, Francesca and Asuman Ozdaglar (2023), \enquote{Graphon games: A
  statistical framework for network games and interventions.}
  \emph{Econometrica}, 91, 191--225.
\bibAnnoteFile{parise2023graphon}

\bibitem[{Rand et~al.(2011)Rand, Arbesman, and Christakis}]{rand2011dynamic}
Rand, David~G, Samuel Arbesman, and Nicholas~A Christakis (2011),
  \enquote{Dynamic social networks promote cooperation in experiments with
  humans.} \emph{Proceedings of the National Academy of Sciences}, 108,
  19193--19198.
\bibAnnoteFile{rand2011dynamic}

\bibitem[{Solimine and Isaac(2023)}]{solimine2023reputation}
Solimine, Philip and R~Mark Isaac (2023), \enquote{Reputation and market
  structure in experimental platforms.} \emph{Journal of Economic Behavior \&
  Organization}, 205, 528--559.
\bibAnnoteFile{solimine2023reputation}

\bibitem[{Wooldridge(2010)}]{wooldridge2010econometric}
Wooldridge, Jeffrey~M (2010), \emph{Econometric analysis of cross section and
  panel data}. MIT press.
\bibAnnoteFile{wooldridge2010econometric}

\end{thebibliography}
\appendix
\includepdf[scale=0.85,pages={1},pagecommand=\section{Experimental instructions\label{app:instructions}}]{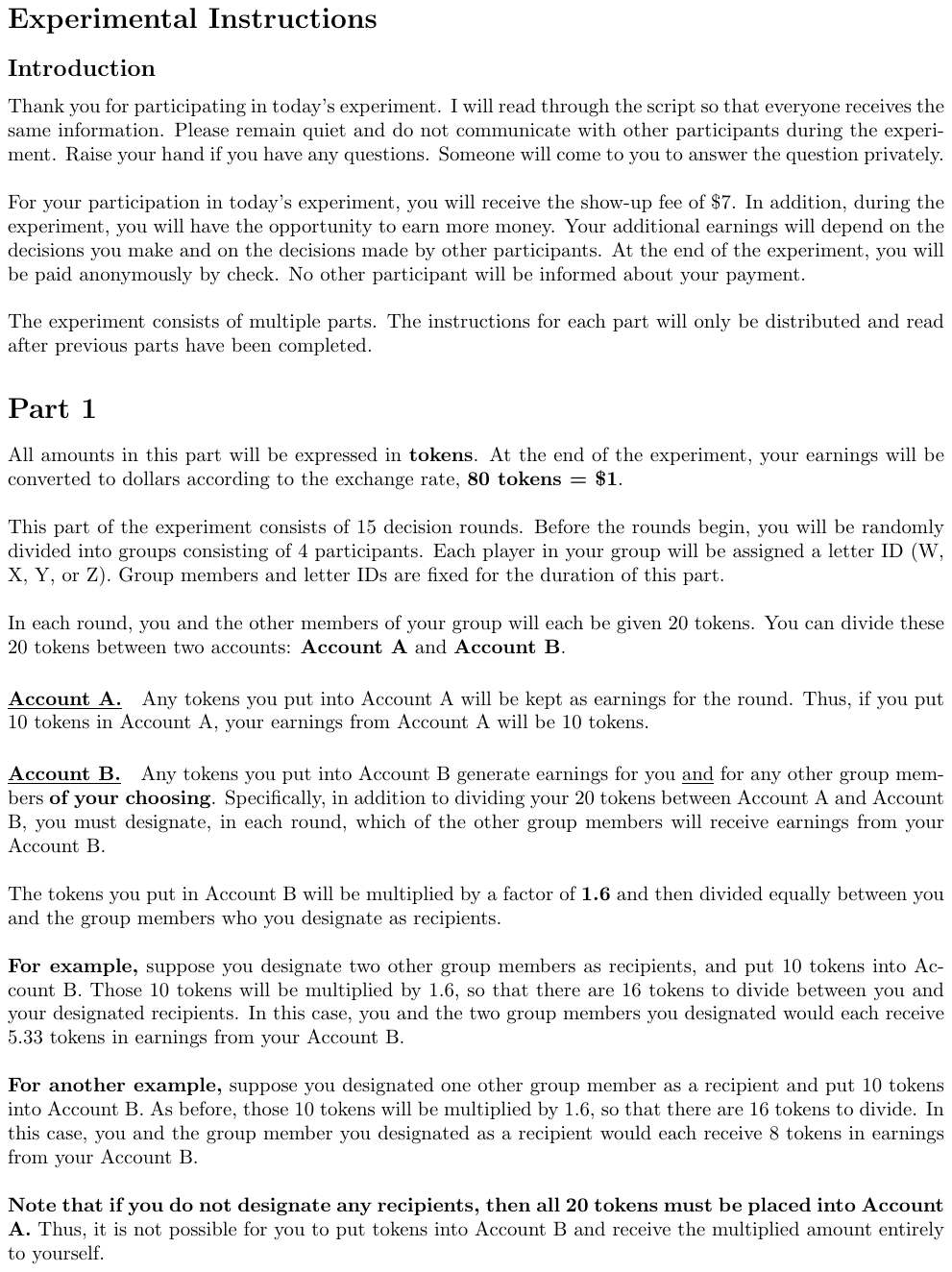}
\includepdf[scale=0.85,pages=2-]{Instructions_091919.pdf}
\includepdf[scale=0.85,pages={1}, pagecommand=\section{Behavioral characteristic questionnaires\label{app:questions}}]{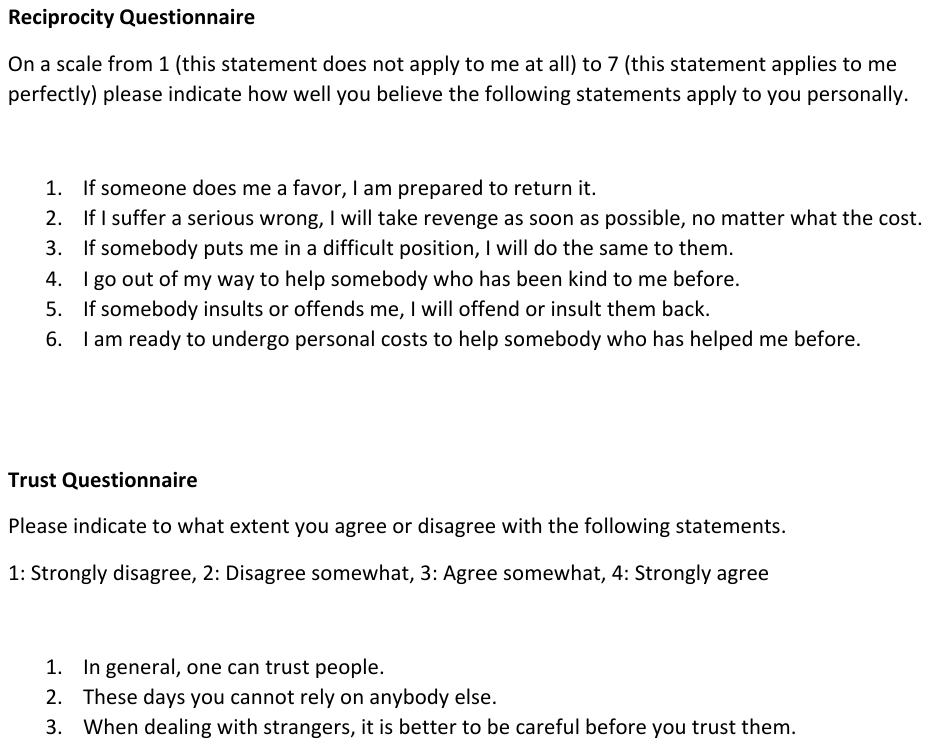}

\end{document}